\documentclass[11pt]{article}

\usepackage{slivkins-setup,slivkins-theorems}

\usepackage{url,fullpage}

\usepackage{kpfonts}
\usepackage[round]{natbib}

\usepackage{booktabs} 
\usepackage{tikz} 

\usepackage[ruled]{algorithm2e} 

\SetAlFnt{\small}
\SetAlCapFnt{\small}
\SetAlCapNameFnt{\small}
\SetAlCapHSkip{0pt}
\IncMargin{-\parindent}

\usepackage[suppress]{color-edits}  
\addauthor{as}{red}      
\addauthor{ym}{cyan}    
\addauthor{sw}{blue}   
\newcommand{\sw}[1]{\swcomment{#1}} 

\DeclareMathOperator*{\Expectation}{\mathbb{E}}
\newcommand{\Ex}[2]{\Expectation_{#1}\left[#2\right]}

\newcommand{\term}[1]{\ensuremath{\mathtt{#1}}\xspace}

\newcommand{\rew}{\term{rew}}  
\newcommand{\PMR}{\term{PMR}} 
\newcommand{\support}{\term{support}}

\newcommand{\BIR}{\term{BIR}} 

\newcommand{\respF}{f_{\term{resp}}}

\newcommand{\HardMax}{\term{HardMax}}
\newcommand{\HardMaxRandom}{\term{HardMax\&Random}}
\newcommand{\SoftMaxRandom}{\term{SoftMax}}
\newcommand{\Uniform}{\term{Uniform}}

\newcommand{\StaticGreedy}{\term{StaticGreedy}}
\newcommand{\DynGreedy}{\term{DynamicGreedy}}

\newcommand{\termSub}[2]{\ensuremath{\mathtt{#1}_{#2}}\xspace}
\newcommand{\alg}[1][]{\termSub{alg}{#1}}

\newcommand{\prior}{\ensuremath{\mP}\xspace}
\newcommand{\priorMu}{\ensuremath{\prior_\mathtt{mean}}\xspace}
\newcommand{\posteriorN}[2]{\mN_{#1,#2}}  

\newcommand{\termTXT}[1]{{\em {#1}}\xspace}

\newcommand{\rationality}{\termTXT{rationality}}

\newcommand{\innovation}{\termTXT{innovation}}

\newcommand{\competition}{\termTXT{competition}}
\newcommand{\Competition}{\termTXT{Competition}}
\newcommand{\competitiveness}{\termTXT{competitiveness}}
\newcommand{\exploration}{\termTXT{exploration}}

\title{Competing bandits: learning under competition%
\footnote{This research has been done while Y. Mansour was co-affiliated with Microsoft Research, and while Z.S. Wu was a graduate student at University of Pennsylvania.}}

\author{Yishay Mansour%
~\thanks{Tel Aviv University.
Email: {\tt  mansour@tau.ac.il}.}
\and
Aleksandrs Slivkins%
~\thanks{Microsoft Research, New York, NY, USA.
Email: {\tt slivkins@microsoft.com}.} \and
Zhiwei Steven Wu%
~\thanks{Microsoft Research, New York, NY, USA.
Email: {\tt zhiww@microsoft.com}.}
}

\date{November 2017}

\begin{document}

\maketitle

\begin{abstract}
Most modern systems strive to learn from interactions with users, and many engage in \emph{exploration}: making potentially suboptimal choices for the sake of acquiring new information. We initiate a study of the interplay between \emph{exploration and competition}---how such systems balance the exploration for learning and the competition for users. Here the users play three distinct roles: they are customers that generate revenue, they are sources of data for learning, and they are self-interested agents which choose among the competing systems.

In our model, we consider competition between two multi-armed bandit algorithms faced with the same bandit instance. Users arrive one by one and choose among the two algorithms, so that each algorithm makes progress if and only if it is chosen.
We ask whether and to what extent competition incentivizes \asedit{the adoption of better bandit algorithms}. We investigate this issue for several models of user response, as we vary the degree of rationality and competitiveness in the model. \asedit{Our findings are closely related to} the ``competition vs. innovation" relationship, a well-studied theme in economics.

\end{abstract}

\section{Introduction}
\label{sec:intro}
Learning from interactions with users is ubiquitous in modern customer-facing systems, from product recommendations to web search to spam detection to content selection to fine-tuning the interface. Many systems purposefully implement \emph{exploration}: making potentially suboptimal choices for the sake of acquiring new information. Randomized controlled trials, a.k.a. A/B testing, are an industry standard, with a number of companies such as \emph{Optimizely} offering tools and platforms to facilitate them. Many companies use more sophisticated exploration methodologies based on \emph{multi-armed bandits}, a well-known theoretical framework for exploration and making decisions under uncertainty.


Systems that engage in exploration typically need to compete against one another; most importantly, they compete for users. This creates an interesting tension between \exploration and \competition. In a nutshell, while exploring may be essential for improving the service tomorrow, it may degrade quality and make users leave \emph{today}, in which case there will be no users to learn from! Thus, users play three distinct roles: they are customers that generate revenue, they generate data for the systems to learn from, and they are self-interested agents which choose among the competing systems.

We initiate a study of the interplay between \exploration and \competition. The main high-level question is: {\bf whether and to what extent competition incentivizes adoption of better exploration algorithms}. This translates into a number of more concrete questions. While it is commonly assumed that better learning technology always helps, is this so for our setting? In other words, would a better learning algorithm result in higher utility for a principal? Would it be used in an equilibrium of the ``competition game"? Also, does competition lead to better social welfare compared to a monopoly? We investigate these questions for several models, as we vary the capacity of users to make rational decisions (\rationality ) and the severity of competition between the learning systems (\competitiveness). The two are controlled by the same ``knob" in our models; \asedit{such coupling is not unusual in the literature, \eg see \citet{Gabaix-16}}.

\asedit{On a high level, our contributions can be framed in terms of the ``inverted-U relationship" between rationality/competitiveness and the quality of adopted algorithms (see \reffig{fig:inverted-U}).}

\OMIT{The relationship between the severity of competition among firms and the quality of technology adopted as a result of this competition is a familiar  theme in the economics literature, known as ``competition vs. innovation".%
\footnote{In this context, \innovation usually refers to adoption of a better technology as a substantial R\&D expense to a given firm. It is not as salient whether similar ideas and/or technologies already exist outside the firm. }
We frame our contributions in terms of the ``inverted-U relationship", a conventional wisdom regarding ``competition vs. innovation" (see \reffig{fig:inverted-U}).}

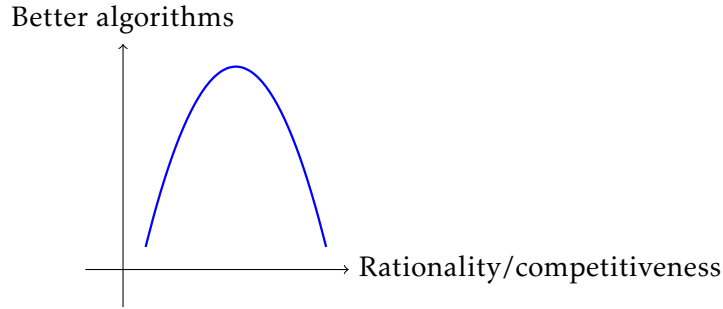
\begin{figure}
\begin{center}
\begin{tikzpicture}
      \draw[->] (-.5,0) -- (3,0) node[right] {Rationality/competitiveness};
      \draw[->] (0,-.5) -- (0,3) node[above] {Better algorithms};
      \draw[scale=0.6,domain=0.5:4.5,smooth,variable=\x,blue, line width=0.3mm] plot ({\x},{4.5 - (\x - 2.5)^2});
 \end{tikzpicture}

\caption{Inverted-U relationship between rationality/competitiveness and algorithms.}
\label{fig:inverted-U}
\end{center}
\end{figure}



\xhdr{Our model.} We define a game in which two firms (\emph{principals}) simultaneously engage in exploration and compete for users (\emph{agents}). These two processes are interlinked, as exploration decisions are experienced by users and informed by their feedback. We need to specify several conceptual pieces: how the principals and agents interact, what is the machine learning problem faced by each principal, and what is the information structure. Each piece can get rather complicated in isolation, let alone jointly, so we strive for simplicity. Thus, the basic model is as follows:

\begin{itemize}

\item A new agent arrives in each round $t=1,2, \ldots$, and chooses among the two principals. The principal chooses an action (\eg a list of web search results to show to the agent), the user experiences this action, and reports a reward. \asedit{All agents have the same ``decision rule" for choosing among the principals given the available information.}

\item Each principal faces a very basic and well-studied version of the multi-armed bandit problem: for each arriving agent, it chooses from a fixed set of actions  (a.k.a. \emph{arms}) and receives a reward drawn independently from a fixed distribution specific to this action.

\item Principals simultaneously announce their learning algorithms \asedit{before round $1$}, and cannot change them afterwards. There is a common Bayesian prior on the rewards (but the realized reward distributions are not observed by the principals or the agents). Agents do not receive \asedit{any other information}. Each principal only observes agents that chose him.
\end{itemize}

\xhdr{Technical results.}
Our results depend crucially on agents' ``decision rule" for choosing among the principals. The simplest and perhaps the most obvious rule is to select the principal which maximizes their expected utility; we refer to it as \HardMax. We find that \HardMax is not conducive to \asedit{adopting better algorithms}. In fact, each principal's dominant strategy is to do no purposeful exploration whatsoever, and instead always choose an action that maximizes expected reward given the current information; we call this algorithm \DynGreedy. While this algorithm may potentially try out different actions over time and acquire useful information, it is known to be dramatically bad in many important cases of multi-armed bandits --- precisely because it does not explore on purpose, and may therefore fail to discover best/better actions. Further, we show that \HardMax is very sensitive to tie-breaking when both principals have exactly the same expected utility according to agents' beliefs. If tie-breaking is probabilistically biased --- say, principal 1 is always chosen with probability strictly larger than $\tfrac12$ --- then this principal has a simple ``winning strategy" no matter what the other principal does.

We relax \HardMax to allow each principal to be chosen with some fixed baseline probability. One intuitive interpretation is that there are ``random agents" who choose a principal uniformly at random, and each arriving agent is either \HardMax or ``random" with some fixed probability. We call this model \HardMaxRandom. We find that \asedit{better algorithms} help in a big way: a sufficiently better algorithm is guaranteed to win all \asedit{non-random} agents after an initial learning phase. While the precise notion of ``sufficiently better algorithm" is rather subtle, we note that commonly known ``smart" bandit algorithms typically defeat the commonly known ``naive" ones, and the latter typically defeat \DynGreedy. However, there is a substantial caveat: one can defeat any algorithm by interleaving it with \DynGreedy. This has two undesirable corollaries: a better algorithm may sometimes lose, and a pure Nash equilibrium typically does not exist.

We further relax the decision rule so that the probability of choosing a given principal varies smoothly as a function of the difference between  principals' expected rewards; we call it \SoftMaxRandom. For this model, the ``better algorithm wins" result holds under much weaker assumptions on what constitutes a better algorithm. This is the most technical result of the paper. The competition in this setting is necessarily much more relaxed: typically, both principals attract approximately half of the agents as time goes by (but a better algorithm may attract slightly more).

All results extend to a much more general version of the multi-armed bandit problem in which the principal may observe additional feedback before and/or after each decision, as long as the feedback distribution does not change over time. In most results, principal's utility may depend on both the market share and agents' rewards.

\xhdr{Economic interpretation.}
\asedit{The inverted-U relationship between the severity of competition among firms and the quality of technologies that they adopt is a familiar theme in the economics literature \citep[\eg][]{Aghion-QJE05,Vives-08}.%
\footnote{The literature frames this relationship as one between ``competition" and ``innovation". In this context, ``innovation" refers to adoption of a better technology, at a substantial R\&D expense to a given firm. It is not salient whether similar ideas and/or technologies already exist outside the firm. It is worth noting that adoption of exploration algorithms tends to require substantial R\&D effort in practice, even if the algorithms themselves are well-known in the research literature; see \citet{MWT-WhitePaper-2016} for an example of such R\&D effort.}
We find it illuminating to frame our contributions in a similar manner, as illustrated in \reffig{fig:inverted-U}.

Our models differ in terms of rationality in agents' decision-making: from fully rational decisions with \HardMax to relaxed rationality with \HardMaxRandom to an even more relaxed rationality with \SoftMaxRandom. The same distinctions also control the severity of competition between the principals: from cut-throat competition with \HardMax to a more relaxed competition with \HardMaxRandom, to an even more relaxed competition with \SoftMaxRandom. Indeed, with \HardMax you lose all customers as soon as you fall behind in performance, with \HardMaxRandom you get some small market share no matter what, and with \SoftMaxRandom you are further guaranteed a market share close to $\tfrac12$ as long as your performance is not much worse than the competition. The uniform choice among principals corresponds to no rationality and no competition.

We identify the inverted-U relationship in the spirit of \reffig{fig:inverted-U} that is driven by the rationality/competitiveness distinctions outlined above: from \HardMax to \HardMaxRandom to \SoftMaxRandom to \Uniform. We also find another, technically different inverted-U relationship which zeroes in on the \HardMaxRandom model. We vary rationality/competitiveness inside this model, and track the marginal utility of switching to a better algorithm.

These inverted-U relationships arise for a fundamentally different reason, compared to the existing literature on ``competition vs. innovation.'' In the literature, better technology always helps in a competitive environment, other things being equal. Thus, the tradeoff is between the costs of improving the technology and the benefits that the improved technology provides in the competition. Meanwhile, we find that a better exploration algorithm may sometimes perform much worse under competition, even in the absence of R\&D costs.}

\ascomment{Yishay's edits, slightly edited by Alex}

\xhdr{Discussion.} We capture some pertinent features of reality while ignoring some others for the sake of tractability. Most notably, we assume that agents do not receive any information about other agents' rewards after the game starts. In the final analysis, this assumption makes agents' behavior independent of a particular realization of the Bayesian prior, and therefore enables us to summarize each learning algorithm via its Bayesian-expected rewards (as opposed to detailed performance on the particular realizations of the prior). Such summarization is essential for formulating lucid and general analytic results, let alone proving them. It is a major open question whether one can incorporate signals about other agents' rewards and obtain a tractable model.

We also make a standard assumption that agents are myopic: they do not worry about how their actions impact their future utility. In particular, they do not attempt to learn over time, to second-guess or game future agents, or to manipulate principal's learning algorithm. We believe this is a typical case in practice, in part because agent's influence tend to be small compared to the overall system. We model this simply by assuming that each agent only arrives once.

Much of the challenge in this paper, both conceptual and technical, was in setting up the right model and the matching theorems, and not only in proving the theorems. Apart from making the modeling choices described above, it was crucial to interpret the results and intuitions from the literature on multi-armed bandits so as to formulate meaningful assumptions on bandit algorithms and Bayesian priors which are productive in our setting.

\xhdr{Open questions.}
\asedit{How to incorporate signals about the other agents' rewards? One needs to reason about how exact or coarse these signals are, and how the agents update their beliefs after receiving them. Also, one may need to allow principals' learning algorithms to respond to updates about the other principal's performance. (Or not, since this is not how learning algorithms are usually designed!) A clean, albeit idealized, model would be that (i) each agent learns her exact expected reward from each principal before she needs to choose which principal to go to, but (ii) these updates are invisible to the principals. Even then, one needs to argue about the competition on particular realizations of the Bayesian prior, which appears very challenging.

Another promising extension is to heterogeneous agents. Then the agents' choices are impacted by their idiosyncratic signals/beliefs, instead of being entirely determined by priors and/or signals about the average performance. It would be particularly interesting to investigate the emergence of \emph{specialization}: whether/when an algorithm learns to target specific population segments in order to compete against a more powerful ``incumbent".
} 

\xhdr{Map of the paper.}
We survey related work (Section~\ref{sec:related-work}), lay out the model and preliminaries (Section~\ref{sec:model}), and proceed to analyze the three main models, \HardMax, \HardMaxRandom and \SoftMaxRandom (in Sections~\ref{sec:rational},~\ref{sec:random}, and ~\ref{sec:soft}, respectively). We discuss economic implications in Section~\ref{sec:welfare}. Appendix~\ref{app:examples} provides some pertinent background on multi-armed bandits. Appendix~\ref{app:perturb} gives a broad example to support an  assumption in our model.


\section{Related work}
\label{sec:related-work}
Multi-armed bandits (\emph{MAB}) is a particularly elegant and tractable abstraction for tradeoff between \emph{exploration} and \emph{exploitation}: essentially, between acquisition and usage of information. MAB problems have been studied in Economics, Operations Research and Computer Science for many decades; see \citep{Bubeck-survey12,Gittins-book11,slivkins-MABbook} for background on regret-minimizing and Bayesian formulations, respectively. A discussion of industrial applications of MAB can be found in \citet{MWT-WhitePaper-2016}.

The literature on MAB is vast and multi-threaded. The most related
thread concerns regret-minimizing MAB formulations with IID rewards
\citep{Lai-Robbins-85,bandits-ucb1}. This thread includes ``smart" MAB
algorithms that combine exploration and exploitation, such as UCB1
\citep{bandits-ucb1} and Successive Elimination
\citep{EvenDar-icml06}, and ``naive'' MAB algorithms that separate
exploration and exploitation, including explore-first and
$\eps$-Greedy \citep[\eg see][]{slivkins-MABbook}.

The three-way tradeoff between exploration, exploitation and incentives has been studied in several other settings:
incentivizing exploration in a recommendation system
    \citep{Che-13,Frazier-ec14,Kremer-JPE14,ICexploration-ec15,Bimpikis-exploration-ms17,Bahar-ec16,ICexplorationGames-ec16-working},
dynamic auctions
    \cite[\eg][]{AtheySegal-econometrica13,DynPivot-econometrica10,Kakade-pivot-or13},
pay-per-click ad auctions with unknown click probabilities
    \cite[\eg][]{MechMAB-ec09,DevanurK09,Transform-ec10-jacm},
coordinating search and matching by self-interested agents
    \citep{Bobby-Glen-ec16},
as well as human computation
    \cite[\eg][]{RepeatedPA-ec14,Ghosh-itcs13,Krause-www13}.

\citet{Bolton-econometrica99,Keller-econometrica05,Johari-ec12} studied models with self-interested agents jointly performing exploration, with no principal to coordinate them.

There is a superficial similarity (in name only) between this paper and the line of work on ``dueling bandits"
    \citep[\eg][]{Yue-dueling12,Yue-dueling-icml09}.
The latter is not about competing bandit algorithms, but rather about scenarios where in each round two arms are chosen to be presented to a user, and the algorithm only observes which arm has ``won the duel".

Our setting is closely related to the ``dueling algorithms" framework \citep{DuelingAlgs-stoc11} which studies competition between two principals, each running an algorithm for the same problem. However, this work considers algorithms for offline / full input scenarios, whereas we focus on online machine learning and the explore-exploit-incentives tradeoff therein. Also, this work specifically assumes binary payoffs (\ie win or lose) for the principals.

\xhdr{Other related work in economics.} The competition vs. innovation relationship and the inverted-U shape thereof have been introduced in a classic book \citep{Schumpeter-42}, and remained an important theme in the literature ever since \cite[\eg][]{Aghion-QJE05,Vives-08}. Production costs aside, this literature treats innovation as a priori beneficial for the firm. Our setting is very different, as innovation in exploration algorithms may potentially hurt the firm.

A line of work on \emph{platform competition}, starting with \cite{Rysman09}, concerns competition between firms (\emph{platforms}) that improve as they attract more users (\emph{network effect}); see \citet{Weyl-White-14} for a recent survey. This literature is not concerned with \innovation, and typically models network effects exogenously, whereas in our model network effects are endogenous: they are created by MAB algorithms, an essential part of the model.

Relaxed versions of rationality similar to ours are found in several notable lines of work. For example, ``random agents" (a.k.a. noise traders) can side-step the ``no-trade theorem'' \citep{Milgrom-Stokey-82}, a famous impossibility result in financial economics. The \SoftMaxRandom model is closely related to the literature on \emph{product differentiation}, starting from \cite{Hotelling-29}, see \cite{Perloff-Salop-85} for a notable later paper.

There is a large literature on non-existence of equilibria due to small deviations   (which is related to the corresponding result for \HardMaxRandom), starting with \cite{Rothschild-Stiglitz-76} in the context of health insurance markets. Notable recent papers \citep{Veiga-Weyl-16,Azevedo-Gottlieb-17} emphasize the distinction between \HardMax and versions of \SoftMaxRandom.

\OMIT{While agents' rationality and severity of competition are often modeled separately in the literature, it is not unusual to have them modeled with the same ``knob" \cite[\eg][]{Gabaix-16}.}


\section{Our model and preliminaries}
\label{sec:model}

\xhdr{Principals and agents.} There are two principals and $T$ agents. The game proceeds in rounds (we will sometimes refer to them as \emph{global rounds}). In each round $t\in [T]$, the following  interaction takes place. A new agent arrives and chooses one of the two principals. The principal chooses a recommendation: an action $a_t\in A$, where $A$ is a fixed set of actions (same for both principals and all rounds). The agent follows this recommendation, receives a reward $r_t\in [0,1]$, and reports it back to the principal.

The rewards are i.i.d. with a common prior. More formally, for each action $a\in A$ there is a parametric family $\psi_a(\cdot)$ of
reward distributions, parameterized by the mean reward $\mu_a$. (The paradigmatic case is 0-1 rewards with a given expectation.) The
mean reward vector $\mu = (\mu_a:\; a\in A)$ is drawn from prior distribution $\priorMu$ before round $1$. Whenever a given action $a\in A$ is chosen, the reward is drawn independently from distribution $\psi_a(\mu_a)$. The prior $\priorMu$ and the distributions $(\psi_a(\cdot)\colon a\in A)$ constitute the (full) Bayesian prior on rewards, denoted $\prior$.

Each principal commits to a learning algorithm for making recommendations. This algorithm follows a protocol of \emph{multi-armed bandits} (\emph{MAB}). Namely, the algorithm proceeds in time-steps:%
\footnote{These time-steps will sometimes be referred to as \emph{local steps/rounds}, so as to distinguish them from ``global rounds" defined before. We will omit the local vs. local distinction when clear from the context.} each time it is called, it outputs a chosen action $a\in A$ and then inputs the reward for this action. The algorithm is called only in global rounds when the corresponding principal is chosen.

The information structure is as follows. The prior $\prior$ is known to everyone. The mean rewards $\mu_a$ are not revealed to anybody. Each agent knows both principals' algorithms, and the global round when (s)he arrives, \emph{but not} the rewards of the previous agents. Each principal is completely unaware of the rounds when the other is chosen.

\xhdr{Some terminology.} The two principals are called ``Principal
1" and ``Principal 2".
The algorithm of principal $i\in \{1,2\}$ is called ``algorithm $i$" and denoted
\alg[i]. The agent in global round $t$ is called ``agent $t$"; the
chosen principal is denoted $i_t$.

Throughout, $\E[\cdot]$ denotes expectation over all applicable randomness.

\xhdr{Bayesian-expected rewards.}
Consider the performance of a given algorithm \alg[i], $i\in \{1,2\}$, when it is run in isolation (\ie without competition, just as a bandit algorithm). Let $\rew_i(n)$ denote its Bayesian-expected reward for the $n$-th step.

Now, going back to our game, fix global round $t$ and let $n_i(t)$ denote the number of global rounds before $t$ in which this principal is chosen. Then:
\begin{align*}
 \E[r_t \mid \text{principal $i$ is chosen in round $t$ and $n_i(t)=n$} ]
    = \rew_i(n+1) \quad (\forall n\in\N).
\end{align*}

\xhdr{Agents' response.}
Each agent $t$ chooses principal $i_t$ as as follows: it chooses a distribution over the principals, and then draws independently from this distribution. Let $p_t$ be the probability of choosing principal $1$ according to this distribution. Below we specify $p_t$; we need to be careful so as to avoid a circular definition.

Let $\mI_t$ be the information available to agent $t$ before the
round. Assume $\mI_t$ suffices to form posteriors for quantities
$n_i(t)$, $i\in \{1,2\}$, denote them by $\posteriorN{i}{t}$. Note
that the Bayesian expected reward of each principal $i$ is a function
only of the number rounds he was chosen by the agents, so the
posterior mean reward for each principal $i$ can be written as
\begin{align*}
 \PMR_i(t) := \E[r_t \mid \text{$\mI_t$ and $i_t=i$} ]
    = \E[\rew_i(n_i(t)+1) \mid \mI_t]
    = \E_{n\sim \posteriorN{i}{t}}[\rew_i(n+1)].
\end{align*}
This quantity represents the posterior mean reward for principal $i$
at round $t$, according to information $\mI_t$; hence the notation \PMR. In general, probability $p_t$ is defined by the
posterior mean rewards $\PMR_i(t)$ for both principals. We assume a
somewhat more specific shape:
\begin{align}
p_t = \respF\left(\; \PMR_1(t) - \PMR_2(t) \;\right).
\end{align}
Here $\respF:[-1,1]\to [0,1]$ is the \emph{response function}, which is the same for all agents. We assume that the response function is known to all agents.

To make the model well-defined, it remains to argue that information $\mI_t$ is indeed sufficient to form posteriors on $n_1(t)$ and $n_2(t)$. This can be easily seen using induction on $t$.

Since all agents arrive with identical information (other than knowing which global round they arrive in), it follows that all agents have identical posteriors for $n_{i,t}$ (for a given principal $i$ and a given global round $t$). This posterior is denoted $\posteriorN{i}{t}$.

\xhdr{Response functions.}
We use the response function $\respF$ to characterize the amount of rationality and competitiveness in our model. We assume that $\respF$ is monotonically non-decreasing, is larger than $\tfrac12$ on the interval $(0,1]$, and smaller than $\tfrac12$ on the interval $[-1,0)$. Beyond that, we consider three specific models, listed in the order of decreasing rationality and competitiveness (see \reffig{fig:response-functions}):

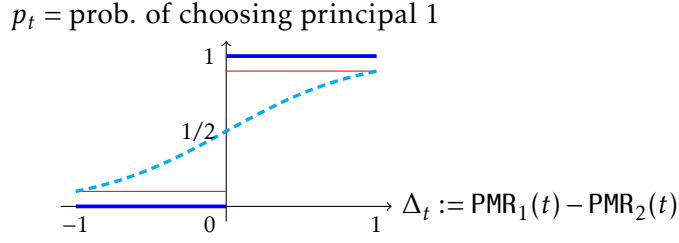
\begin{figure}
\begin{center}
  \begin{tikzpicture}[scale=2.0]
    \draw[->] (-1.1,0) -- (1.1,0) node[right] {$\Delta_t := \PMR_1(t) - \PMR_2(t)$};
    \draw[->] (0,-0.1) -- (0,1.1) node[above]
        {$p_t = \text{prob. of choosing principal 1}$};
    \draw[scale=1.0,domain=-1:0,smooth,variable=\q,blue, line width=0.50mm] plot ({\q},{0});
    \draw[scale=1.0,domain=0:1,smooth,variable=\q,blue,line width=0.50mm] plot ({\q},{1});
    \draw[scale=1.0,domain=-1:0,smooth,variable=\y,red]  plot ({\y},{0.1});
    \draw[scale=1.0,domain=0:1,smooth,variable=\y,red]  plot ({\y},{0.9});
    \draw[scale=1.0,domain=-1:1,smooth,variable=\y,cyan, line width=0.45mm, dash pattern=on 3pt off 2pt]  plot ({\y},{1/(1 + 1/(9^\y))});
     \node[left] at (0, 0.5) {\footnotesize $1/2$};
     \node[left] at (0, 1) {\footnotesize $1$};
     \node[below left] at (0, 0) {\footnotesize $0$};
     \node[below ] at (1, 0) {\footnotesize $1$};
     \node[below ] at (-1, 0) {\footnotesize $-1$};
  \end{tikzpicture}
\end{center}
\caption{The three models for agents' response function: \HardMax is thick blue, \HardMaxRandom is slim red, and \SoftMaxRandom is the dashed curve.}
\label{fig:response-functions}
\end{figure}

\begin{OneLiners}
\item \HardMax: $\respF$ equals $0$ on the interval $[-1,0)$ and $1$
  on the interval $(0,1]$. In other words, the agents will
  deterministically choose the principal with the higher posterior
  mean reward.

\item \HardMaxRandom:
    $\respF$ equals $\eps_0$ on the interval $[-1,0)$ and $1-\eps_0$ on the interval $(0,1]$, where $\eps_0\in (0,\tfrac12)$ are some positive constants. In words, each agent is a \HardMax agent with probability $1-2\eps_0$, and with the remaining probability she makes a random choice.

\item \SoftMaxRandom: $\respF(\cdot)$ lies in the interval $[\eps_0,1-\eps_0]$, $\eps_0>0$, and is ``smooth" around $0$ (in the sense defined precisely in Section~\ref{sec:soft}).
\end{OneLiners}

We say that $\respF$ is \emph{symmetric} if $\respF(-x)+\respF(x)=1$
for any $x\in [0,1]$. This implies \emph{fair tie-breaking}:
$\respF(0)=\tfrac12$.

\xhdr{MAB algorithms.}
We characterize the inherent quality of an MAB algorithm in terms of its \emph{Bayesian Instantaneous Regret} (henceforth, \BIR), a standard notion from machine learning:
\begin{align}
\BIR(n) := \E_{\mu\sim\priorMu}
    \left[ \max_{a\in A} \mu_a\right] - \rew(n),
\end{align}
where $\rew(n)$ is the Bayesian-expected reward of the algorithm for the $n$-th step, when the algorithm is run in isolation. We are primarily interested in how \BIR scales with $n$; we treat $K$, the number of arms, as a constant unless specified otherwise.

We will emphasize several specific algorithms or classes thereof:

\begin{itemize}
\item ``smart" MAB algorithms that combine exploration and exploitation, such as UCB1 \cite{bandits-ucb1} and Successive Elimination \cite{EvenDar-icml06}. These algorithms achieve
        $\BIR(n) \leq  \tilde{O}(n^{-1/2})$
    for all priors and all (or all but a very few) steps $n$. This bound is known to be tight for any fixed $n$.
    \footnote{This follows from the lower-bound analysis in \cite{bandits-exp3}.}

\item ``naive" MAB algorithms that separate exploration and exploitation, such as Explore-then-Exploit and $\eps$-Greedy. These algorithms have dedicated rounds in which they explore by choosing an action uniformly at random. When these rounds are known in advance, the algorithm suffers constant \BIR in such rounds. When the ``exploration rounds" are instead randomly chosen by the algorithm, one can usually guarantee an inverse-polynomial upper bound \BIR, but not as good as the one above: namely,
        $ \BIR(n) \leq \tilde{O}(n^{-1/3})$.
    This is the best possible upper bound on \BIR for the two algorithms mentioned above.

\item \DynGreedy: at each step, recommends the best action according to the current posterior: an action $a$ with the highest posterior expected reward
        $\E[\mu_a \mid \mI]$,
     where $\mI$ is the information available to the algorithm so far.
     \DynGreedy has (at least) a constant \BIR for some reasonable priors, \ie
        $\BIR(n)>\Omega(1)$.

\item \StaticGreedy: always recommends the prior best action,\ie an action $a$ with the highest prior mean reward  $\E_{\mu\sim \priorMu}[\mu_a]$. This algorithm typically has constant \BIR.
\end{itemize}

We focus on MAB algorithms such that $\BIR(n)$ is non-increasing; we call such algorithms \emph{monotone}. While some reasonable MAB algorithms may occasionally violate monotonicity, they can usually be easily modified so that monotonicity violations either vanish altogether, or only occur at very specific rounds (so that agents are extremely unlikely to exploit them in practice).

More background and examples can be found in Appendix~\ref{app:examples}. In particular, we prove that \DynGreedy is monotone.

\xhdr{Competition game between principals.}
Some of our results explicitly study the game between the two principals. We model it as a simultaneous-move game: before the first agent arrives, each principal commits to an MAB algorithm. Thus, choosing a pure strategy in this game corresponds to choosing an MAB algorithm (and, implicitly, announcing this algorithm to the agents).

Principal's utility is primarily defined as the market share, \ie the number of agents that chose this principal. Principals are risk-neutral, in the sense that they optimize their expected utility.

\xhdr{Assumptions on the prior.} We make some technical assumptions for the sake of simplicity. First, each action $a$ has a positive probability of being the best action according to the prior:
\begin{align}\label{eq:assn-prob}
\forall a\in A:\;\;\;  \Pr_{\mu\sim \priorMu}[\mu_a  > \mu_{a'}
\;\forall a'\in A]
> 0.
\end{align}

Second, posterior mean rewards of actions are pairwise distinct almost
surely. That is, the history $h$ at any step of an MAB algorithm%
\footnote{The \emph{history} of an MAB algorithm at a given step
  comprises the chosen actions and the observed rewards in all
  previous steps in the execution of this algorithm.}
 satisfies
\begin{align}\label{eq:assn-distinct}
    \E[\mu_a \mid h] \neq \E[\mu_{a'}\mid h] \quad \forall a,a'\in A,
\end{align}
\asedit{except at a set of histories of probability $0$}.  In
particular, prior mean rewards of actions are pairwise distinct:
$\E[\mu_a] \neq \E[\mu_a']$ for any $a,a'\in A$. 

We provide two examples for which property \eqref{eq:assn-distinct} is `generic', in the sense that it can be enforced almost surely by a
  small random perturbation of the prior. Both examples focus on 0-1 rewards and priors $\priorMu$ that are independent across arms. The first example
  assumes Beta priors on the mean rewards, and is very easy.%
  \footnote{Suppose the rewards are Bernouli r.v. and the mean reward
    $\mu_a$ for each arm $a$ is drawn from some Beta distribution
    $\text{Beta}(\alpha_a, \beta_a)$. Given any history that contains
    $h_a$ number of heads and $t_a$ number of tails from arm $a$, the
    posterior mean reward is
    $\frac{\alpha_a + h_a}{\alpha_a + h_a + \beta_a + t_a}$. Note that
    $h_a$ and $t_a$ take integer values. Therefore, perturbing the
    parameters $\alpha_a$ and $\beta_a$ independently with any
    continuous noise will induce a prior with property
    \eqref{eq:assn-distinct} with probability 1.  }  The second
  example assumes that mean rewards have a
  finite support, see Appendix~\ref{app:perturb} for details.

\xhdr{Some more notation.} Without loss of generality, we label actions as $A=[K]$ and sort them according to their prior mean rewards, so that
    $ \E[\mu_1] > \E[\mu_2] > \ldots > \E[\mu_K]$.

Fix principal $i\in \{1,2\}$ and (local) step $n$. The arm chosen by algorithm \alg[i] at this step is denoted $a_{i,n}$, and the corresponding \BIR is denoted $\BIR_i(n)$. History of \alg[i] up to this step is denoted $H_{i,n}$.

Write
    $\PMR(a\mid E) = \E[\mu_a \mid E]$
for posterior mean reward of action $a$ given event $E$.


\subsection{Generalizations}
\label{sec:model-extensions}

Our results can be extended compared to the basic model described above.


First, unless specified otherwise, our results allow a more general notion of principal's utility that can depend on both the market share and agents' rewards. Namely, principal $i$ collects $U_i(r_t)$ units of utility in each global round $t$ when she is chosen (and $0$ otherwise), where $U_i(\cdot)$ is some fixed non-decreasing function with $U_i(0)>0$. In a formula,
\begin{align}\label{eq:general-utility}
\textstyle U_i := \sum_{t=1}^T\; \indicator{i_t=i}\cdot U_i(r_r).
\end{align}

Second, our results carry over, with little or no modification of the proofs, to much more general versions of MAB, as long as it satisfies the i.i.d. property. In each round, an algorithm can see a \emph{context} before choosing an action (as in \emph{contextual bandits}) and/or additional feedback other than the reward after the reward is chosen (as in, \eg \emph{semi-bandits}), as long as the contexts are drawn from a fixed distribution, and the (reward, feedback) pair is drawn from a fixed distribution that depends only on the context and the chosen action. The Bayesian prior $\prior$ needs to be a more complicated object, to make sure that \PMR and \BIR are well-defined. Mean rewards may also have a known structure, such as Lipschitzness, convexity, or linearity; such structure can be incorporated via $\prior$. All these extensions have been studied extensively in the literature on MAB, and account for a substantial segment thereof; see \cite{Bubeck-survey12} for background and details.

\subsection{Chernoff Bounds}

We use an elementary concentration inequality known as {\em Chernoff Bounds}, in a formulation from~\cite{MitzUpfal-book05}.
\begin{theorem}[Chernoff Bounds]
\label{thm:chernoff}
Consider $n$ i.i.d. random variables $X_1 \ldots X_n$ with values in $[0,1]$. Let
    $X = \tfrac{1}{n} \sum_{i=1}^n X_i$ be their average, and let $\nu = \E[X]$. Then:
\[ \min\left(\; \Pr[ X-\nu > \delta \nu ],\quad
                \Pr[ \nu-X > \delta \nu ]
    \; \right)
    < e^{-\nu n \delta^2/3}
    \quad \text{for any $\delta\in (0,1)$.}
\]
\end{theorem}

\section{Full rationality (HardMax)}
\label{sec:rational}

In this section, we will consider the version in which
the agents are fully rational, in the sense that their response function is \HardMax. We show that principals are not incentivized to \emph{explore}--- \ie to deviate from \DynGreedy. The core technical result is that if one principal adopts \DynGreedy, then the other principal loses all agents as soon as he deviates.

To make this more precise, let us say that two MAB algorithms \emph{deviate} at (local) step $n$ if there is an action $a\in A$ and \asedit{a set of step-$n$ histories of positive probability such that any history $h$ in this set} is feasible for both algorithms, and under this history the two algorithms choose action $a$ with different probability.

\begin{theorem}\label{thm:DG-dominance}
Assume \HardMax response function with fair tie-breaking. Assume that \alg[1] is \DynGreedy, and \alg[2] deviates from \DynGreedy starting from some (local) step $n_0<T$. Then all agents in global rounds $t\geq n_0$ select principal $1$.
\end{theorem}

\begin{corollary}\label{cor:DG-dominance}
The competition game between principals has a unique Nash equilibirium: both principals choose \DynGreedy.
\end{corollary}

\begin{remark}
This corollary holds under a more general model which allows time-discounting: namely, the utility of each principal $i$ in each global round $t$ is $U_{i,t}(r_t)$ if this principal is chosen, and $0$ otherwise, where $U_{i,t}(\cdot)$ is an arbitrary non-decreasing function with $U_{i,t}(0)>0$.
\end{remark}

\subsection{Proof of Theorem~\ref{thm:DG-dominance}}

The proof starts with two auxiliary lemmas: that deviating from \DynGreedy implies a strictly smaller Bayesian-expected reward, and that \HardMax implies a ``sudden-death" property: if one agent chooses principal $1$ with certainty, then so do all subsequent agents do. \asedit{We re-use both lemmas in later sections, so we state them in sufficient generality.}

\begin{lemma}\label{lm:DG-rew}
\asedit{Assume that \alg[1] is \DynGreedy, and \alg[2] deviates from \DynGreedy starting from some (local) step $n_0<T$. Then $\rew_1(n_0)>\rew_2(n_0)$. This holds for any response function $\respF$.}
\end{lemma}

\asedit{Lemma~\ref{lm:DG-rew} does not rely on any particular shape of the response function because it only considers the performance of each algorithm without competition.}

\begin{proof}[Proof of Lemma~\ref{lm:DG-rew}]
Since the two algorithms coincide on the first $n_0-1$ steps, it follows by symmetry that histories $H_{1,n_0}$ and $H_{2,n_0}$ have the same distribution. We use a \emph{coupling argument}: w.l.o.g., we assume the two histories coincide,
    $H_{1,n_0} = H_{2,n_0} = H$.

At local step $n_0$, \DynGreedy chooses an action $a_{1,n_0} = a_{1,n_0}(H)$
which maximizes the posterior mean reward given history $H$: for any realized history $h\in \support(H)$ and any action $a\in A$
\begin{align}\label{eq:lm:DG-rew-1}
 \PMR(a_{1,n_0} \mid H = h) \geq \PMR(a \mid H=h).
\end{align}

\ascomment{Rewrote the rest of the proof to account for positive-prob set of histories.}

By assumption \eqref{eq:assn-distinct}, it follows that 
\begin{align}\label{eq:lm:DG-rew-2}
 \PMR(a_{1,n_0} \mid H = h) > \PMR(a \mid H=h)
 \quad \text{for any $h\in \support(H)$ and $a\neq a_{1,n_0}(h)$.}
\end{align}

Since the two algorithms deviate at step $n_0$, there is a set $S\subset \support(H)$ of step-$n_0$ histories such that $\Pr[S]>0$ and any history $h\in S$ satisfies
    $\Pr[a_{2,n_0}\neq a_{1,n_0} \mid H=h]>0$.
Combining this with \eqref{eq:lm:DG-rew-2}, we deduce that 
\begin{align}\label{eq:lm:DG-rew-3}
 \PMR(a_{1,n_0} \mid H = h) > \E\left[ \mu_{a_{2,n_0}}\mid H=h \right]
 \quad\text{for each history $h\in S$}.
\end{align}
Using \eqref{eq:lm:DG-rew-1} and \eqref{eq:lm:DG-rew-3} and integrating over realized histories $h$, we obtain
    $\rew_1(n_0) > \rew_2(n_0)$.
\end{proof}

\begin{lemma}\label{lm:DG-sudden}
\asedit{Consider \HardMax response function with $\respF(0)\geq\tfrac12$.} 
Suppose \alg[1] is monotone, and $\PMR_1(t_0)>\PMR_2(t_0)$ for some global round $t_0$. Then $\PMR_1(t)>\PMR_2(t)$ for all subsequent rounds $t$.
\end{lemma}

\begin{proof}
Let us use induction on round $t\geq t_0$, with the base case $t=t_0$. Let $\mN=\mN_{1,t_0}$ be the agents' posterior distribution for $n_{1,t_0}$, the number of global rounds before $t_0$ in which principal $1$ is chosen. By induction, all agents from $t_0$ to $t-1$ chose principal $1$, so $\PMR_2(t_0)= \PMR_2(t)$. Therefore,
\[ \PMR_1(t)
    = \Ex{n\sim \mN}{\rew_1(n+1+t-t_0)}
    \geq \Ex{n\sim \mN}{\rew_1(n+1)}
    =\PMR_1(t_0) > \PMR_2(t_0)= \PMR_2(t), \]
where the first inequality holds because \alg[1] is monotone, and the second one is the base case.
\end{proof}

\begin{proof}[Proof of Theorem~\ref{thm:DG-dominance}]
Since the two algorithms coincide on the first $n_0-1$ steps, it follows by symmetry that $\rew_1(n) = \rew_2(n)$ for any $n< n_0$.
By Lemma~\ref{lm:DG-rew},
    $\rew_1(n_0) > \rew_2(n_0)$.

Recall that $n_i(t)$ is the number of global rounds $s<t$ in which principal $i$ is chosen, and $\posteriorN{i}{t}$ is the agents' posterior distribution for this quantity. By symmetry, each agent $t<n_0$ chooses a principal uniformly at random. It follows that
    $\posteriorN{1}{n_0} = \posteriorN{2}{n_0}$
(denote both distributions by $\mN$ for brevity), and $\mN(n_0-1)>0$.
Therefore:
\begin{align}
\PMR_1(n_0)
  &= \Ex{n\sim \mN} {\rew_1(n+ 1)}
   = \sum_{n = 0}^{n_0-1} \mN(n) \cdot{\rew_1(n + 1)} \nonumber \\
  & > \mN(n_0-1)\cdot {\rew_2(n_0)} + \sum_{n = 0}^{n_0-2}  \mN(n)\cdot{\rew_2(n + 1)}
    \nonumber \\
  &= \Ex{n\sim \mN}{\rew_2(n + 1)} = \PMR_2(n_0) \label{eq:pf:thm:DG-dominance}
\end{align}
So, agent $n_0$ chooses principal $1$. By Lemma~\ref{lm:DG-sudden} \asedit{(noting that \DynGreedy is monotone)}, all subsequent agents choose principal $1$, too.
\end{proof}

%
%

\subsection{\HardMax with biased tie-breaking}
\label{sec:HardMax-biased}

The \HardMax model is very sensitive to the tie-breaking rule. For starters, if ties are  broken deterministically in favor of principal $1$, then principal 1 can get all agents no matter what the other principal does, simply by using \StaticGreedy.

\begin{theorem}\label{thm:HardMax-hardTies}
Assume \HardMax response function with $\respF(0)=1$ (ties are always broken in favor of principal $1$). If \alg[1] is \StaticGreedy, then all agents choose principal $1$.
\end{theorem}

\begin{proof}
Agent $1$ chooses principal $1$ because of the tie-breaking rule. Since \StaticGreedy is trivially monotone, all the subsequent agents choose principal $1$ by an induction argument similar to the one in the proof of Lemma~\ref{lm:DG-sudden}.
\end{proof}

A more challenging scenario is when the tie-breaking is biased in favor of principal 1, but not deterministically so: $\respF(0)>\tfrac12$. Then this principal also has a ``winning strategy" no matter what the other principal does. Specifically, principal 1 can get all but the first few agents, under a mild technical assumption that \DynGreedy deviates from \StaticGreedy. Principal 1 can use \DynGreedy, or any other monotone MAB algorithm that coincides with \DynGreedy in the first few steps.


\begin{theorem}\label{thm:HardMax-biased}
Assume \HardMax response function with $\respF(0)>\tfrac12$ (\ie tie-breaking is biased in favor of principal $1$). Assume the prior $\mP$ is such that \DynGreedy deviates from \StaticGreedy starting from some step $n_0$. Suppose that principal $1$ runs a monotone MAB algorithm that coincides with \DynGreedy in the first $n_0$ steps. Then all agents $t\geq n_0$ choose principal $1$.
\end{theorem}

\begin{proof}
The proof re-uses Lemmas~\ref{lm:DG-rew} and~\ref{lm:DG-sudden}, which do not rely on fair tie-breaking.

Because of the biased tie-breaking, for each global round $t$ we have:
\begin{align}\label{eq:thm:HardMax-biased-PMRtoPr}
\text{if $\PMR_1(t)\geq \PMR_2(t)$ then $\Pr[i_t=1]>\tfrac12$.}
\end{align}
Recall that $i_t$ is the principal chosen in global round $t$.

Let $m_0$ be the first step when \alg[2] deviates from \DynGreedy, or \DynGreedy deviates from \StaticGreedy, whichever comes sooner. \asedit{Then \alg[2], \DynGreedy and \StaticGreedy coincide on the first $m_0-1$ steps. Moreover, $m_0\leq n_0$ (since \DynGreedy deviates from \StaticGreedy at step $n_0$), so \alg[1] coincides with \DynGreedy on the first $m_0$ steps.

So, $\rew_1(n)=\rew_2(n)$ for each step $n<m_0$, because \alg[1] and \alg[2] coincide on the first $m_0-1$ steps. Moreover, if \alg[2] deviates from \DynGreedy at step $m_0$ then
    $\rew_1(m_0) > \rew_2(m_0)$ 
by Lemma~\ref{lm:DG-rew}; else, we trivially have
    $\rew_1(m_0) = \rew_2(m_0)$.} To summarize:
\begin{align}\label{eq:thm:HardMax-biased-rew}
    \rew_1(n)\geq\rew_2(n)\quad \text{for all steps $n\leq m_0$}.
\end{align}

We claim that $\Pr[i_t=1]>\tfrac12$ for all global rounds $t\leq m_0$. We prove this claim using induction on $t$. The base case $t=1$ holds by \eqref{eq:thm:HardMax-biased-PMRtoPr} and the fact that in step 1, \DynGreedy chooses the arm with the highest prior mean reward. For the induction step, we assume that $\Pr[i_t=1]>\tfrac12$ for all global rounds $t<t_0$, for some $t_0\leq  m_0$. It follows that distribution $\posteriorN{1}{t_0}$ stochastically dominates distribution $\posteriorN{2}{t_0}$.%
\footnote{For random variables $X,Y$ on \R, we say that $X$ \emph{stochastically dominates} $Y$ if $\Pr[X\geq x] \geq \Pr[Y\geq x]$ for any $x\in \R$.}
Observe that
\begin{align}\label{eq:thm:HardMax-biased-PMR-aux}
\PMR_1(t_0)
  = \Ex{n\sim \posteriorN{1}{t_0}} {\rew_1(n+1)}
  \geq \Ex{n\sim \posteriorN{2}{t_0}} {\rew_2(n+1)}
  = \PMR_2(t_0).
\end{align}
So the induction step follows by \eqref{eq:thm:HardMax-biased-PMRtoPr}. Claim proved.

Now let us focus on global round $m_0$, and denote $\mN_i = \posteriorN{i}{m_0}$.  By the above claim,
\begin{align}\label{eq:thm:HardMax-biased-mN}
\text{$\mN_1$ stochastically dominates $\mN_2$, and moreover
    $\mN_i(m_0-1)>\mN_i(m_0-1)$.}
\end{align}

By definition of $m_0$, either (i) \alg[2] deviates from \DynGreedy starting from local step $m_0$, which implies $\rew_1(m_0)> \rew_2(m_0)$ by Lemma~\ref{lm:DG-rew}, or (ii) \DynGreedy deviates from \StaticGreedy starting from local step $m_0$, which implies $\rew_1(m_0)>\rew_1(m_0-1)$ by Lemma~\ref{dgmono}. In both cases, using \eqref{eq:thm:HardMax-biased-rew} and \eqref{eq:thm:HardMax-biased-mN}, it follows that the inequality in \eqref{eq:thm:HardMax-biased-PMR-aux} is strict for $t_0=m_0$.

Therefore, agent $m_0$ chooses principal $1$, and by Lemma~\ref{lm:DG-sudden} so do all subsequent agents.
\end{proof}

\section{Relaxed rationality: HardMax \& Random}
\label{sec:random}
This section is dedicated to the \HardMaxRandom response model, where each principal is always chosen with some positive baseline probability. The main technical result for this model states that a principal with asymptotically better \BIR wins by a large margin: after a ``learning phase" of constant duration, all agents choose this principal with maximal possible probability $\respF(1)$. For example, a principal with $\BIR(n)\leq \tilde{O}(n^{-1/2})$ wins over a principal with $\BIR(n)\geq \Omega(n^{-1/3})$. However, this positive result comes with a significant caveat detailed in Section~\ref{sec:random-greedy}.

We formulate and prove a cleaner version of the result, followed by a more general formulation developed in a subsequent Remark~\ref{rem:random-messy}. We need to express a property that \alg[1] eventually catches up and surpasses \alg[2], even if initially it receives only a fraction of traffic. For the cleaner version, we assume that both algorithms are well-defined for an infinite time horizon, so that their \BIR does not depend on the time horizon $T$ of the game.  Then this property can be formalized as:
\begin{align}\label{eq:random-better-clean}
(\forall \eps>0)\qquad
\BIR_1(\eps n)/\BIR_2(n) \to 0.
\end{align}
In fact, a weaker version of \eqref{eq:random-better-clean} suffices:
denoting $\eps_0 = \respF(-1)$, for some constant $n_0$ we have
\begin{align}\label{eq:random-better-weaker}
(\forall n\geq n_0) \qquad
\BIR_1(\eps_0 n/2)/\BIR_2(n) <\tfrac12.
\end{align}
We also need a very mild technical assumption on the ``bad" algorithm:
\begin{align}\label{eq:random-assn}
 (\forall n\geq n_0) \qquad
  \BIR_2(n) > 4\,e^{-\eps_0 n/12}.
\end{align}

\begin{theorem}\label{thm:random-clean}
Assume \HardMaxRandom response function. Suppose both algorithms are monotone and well-defined for an infinite time horizon, and satisfy~\eqref{eq:random-better-weaker} and~\eqref{eq:random-assn}. Then each agent $t\geq n_0$ chooses principal $1$ with maximal possible probability $\respF(1) = 1- \eps_0$.
\end{theorem}

\begin{proof}
Consider global round $t\geq n_0$. Recall that each agent chooses principal $1$ with probability at least
    $\respF(-1)>0$.

Then
    $\E[n_1(t+1)] \geq 2\eps_0\,t $.
By Chernoff Bounds (Theorem~\ref{thm:chernoff}), we have that
    $n_1(t+1)\geq \eps_0 t$
holds with probability at least $1-q$,
where $q = \exp(-\eps_0 t/12)$.

We need to prove that
    $\PMR_1(t) - \PMR_2(t)>0$.
For any $m_1$ and $m_2$, consider the quantity
\[ \Delta(m_1,m_2) := \BIR_2(m_2+1) - \BIR_1(m_1+1).\]
Whenever $m_1 \geq \eps_0 t/2 -1$ and $m_2<t$, it holds that
\begin{align*}
\Delta(m_1,m_2) \geq \Delta(\eps_0 t / 2, t)
    \geq \BIR_2(t)/2.
\end{align*}
The above inequalities follow, resp., from algorithms' monotonicity and \eqref{eq:random-better-weaker}. Now,
\begin{align*}
\PMR_1(t) - \PMR_2(t)
    &= \Ex{m_1\sim \posteriorN{1}{t},\;m_2\sim \posteriorN{2}{t}}{\Delta(m_1,m_2)} \\
    &\geq -q
        + \Ex{m_1\sim \posteriorN{1}{t},\;m_2\sim \posteriorN{2}{t}}
            {\Delta(m_1,m_2)\mid m_1 \geq \eps_0 t/2-1} \\
    &\geq \BIR_2(t)/2-q \\
    &> \BIR_2(t)/4 > 0  \EqComment{by \eqref{eq:random-assn}}. \qedhere
\end{align*}
\end{proof}

\begin{remark}\label{rem:random-messy}
  Many standard MAB algorithms in the literature are parameterized by
  the time horizon $T$. Regret bounds for such algorithms usually
  include a polylogarithmic dependence on $T$. In particular, a
  typical upper bound for \BIR has the following form:
  \OMIT{\footnote{We
    provide upper bounds on \BIR for several standard MAB algorithms
    to illustrate these dependencies in the appendix.\sw{added}}}
\begin{align}
    \BIR(n\mid T)\leq \polylog(T)\cdot n^{-\gamma}
    \quad \text{for some $\gamma\in(0,\tfrac12]$}.
\end{align}
Here we write $\BIR(n\mid T)$ to emphasize the dependence on $T$.

We generalize \eqref{eq:random-better-weaker} to handle the dependence
on $T$: there exists a number $T_0$ and a function $n_0(T)\in \polylog(T)$ 
such that 
\begin{align}\label{eq:random-better-messy}
(\forall T\geq T_0,\; n\geq n_0(T)) \quad
\frac{\BIR_1(\eps_0 n /2\mid T)}{\BIR_2(n\mid T)} <\frac12.
\end{align}
If this holds, we say that \alg[1] \emph{BIR-dominates} \alg[2].

We provide a version of Theorem~\ref{thm:random-clean} in which algorithms are parameterized with time horizon $T$ and condition \eqref{eq:random-better-weaker} is replaced with \eqref{eq:random-better-messy}; its proof is very similar and is omitted.
\end{remark}



To state a game-theoretic corollary of Theorem~\ref{thm:random-clean}, we consider a version of the competition game between the two principals in which they can only choose from a finite set $\mA$ of monotone MAB algorithms. One of these algorithms is ``better" than all others; we call it the \emph{special} algorithm. Unless specified otherwise, it BIR-dominates all other allowed algorithms. The other algorithms satisfy \eqref{eq:random-assn}. We call this game the \emph{restricted competition game}.

\begin{corollary}\label{cor:random}
Assume \HardMaxRandom response function. Consider the restricted competition game with special algorithm \alg. Then, for any sufficiently large time horizon $T$, this game has a unique Nash equilibrium: both principals choose \alg.
\end{corollary}

\subsection{A little greedy goes a long way}
\label{sec:random-greedy}

Given any monotone MAB algorithm other than \DynGreedy, we design a modified algorithm which learns at a slower rate, yet ``wins the game" in the sense of Theorem~\ref{thm:random-clean}. As a corollary, the competition game with unrestricted choice of algorithms typically does not have a Nash equilibrium.

Given an algorithm \alg[1] that deviates from \DynGreedy starting from
step $n_0$ and a ``mixing'' parameter $p$, we will construct a
modified algorithm as follows.
\begin{enumerate}
\item The modified algorithm coincides with \alg[1] (and \DynGreedy)
for the first $n_0-1$ steps;
\item In each step $n\geq n_0$, \alg[1] is invoked with probability
  $1-p$, and with the remaining probability $p$ does the ``greedy
  choice": chooses an action with the largest posterior mean reward
  given the current information collected by \alg[1].
\end{enumerate}
For a cleaner comparison between the two algorithms, the modified algorithm does not record rewards received in steps with the ``greedy choice". Parameter $p>0$ is the same for all steps.

\begin{theorem}\label{thm:random-greedy}
Assume symmetric \HardMaxRandom response function. Let $\eps_0 = \respF(-1)$ be the baseline probability. Suppose \alg[1] deviates from \DynGreedy starting from some step $n_0$. Let \alg[2] be the modified algorithm, as described above, with mixing parameter $p$ such that
    $(1-\eps_0)(1-p)>\eps_0$.
Then each agent $t\geq n_0$ chooses principal $2$ with maximal possible probability $1-\eps_0$.
\end{theorem}

\begin{corollary}\label{cor:random-greedy}
  Suppose that both principals can choose any monotone MAB algorithm, and assume the symmetric \HardMaxRandom response
  function. Then for any time
  horizon $T$, the only possible \emph{pure} Nash equilibrium is one
  where both principals choose \DynGreedy. Moreover, no pure Nash
  equilibrium exists when some algorithm ``dominates" \DynGreedy in
  the sense of \eqref{eq:random-better-messy} and the time horizon $T$
  is sufficiently large.
\end{corollary}

\begin{remark}
The modified algorithm performs exploration 
at a slower rate. Let us argue how this may translate  into a larger \BIR compared to the original algorithm. Let  $\BIR'_1(n)$ be the \BIR of the ``greedy choice" after after $n-1$ steps of \alg[1]. Then
\begin{align}\label{eq:random-BIR}
\BIR_2(n)
    &= \Ex{m\sim (n_0-1)+\term{Binomial}(n-n_0+1,1-p)}{(1-p) \cdot \BIR_1(m) + p \cdot \BIR'_1(m)}.
\end{align}
In this expression, $m$ is the number of times \alg[1] is invoked in the first $n$ steps of the modified algorithm. Note that
    $\E[m] = n_0-1 + (n-n_0+1)(1-p) \geq (1-p)n$.

Suppose $\BIR_1(n)= \beta n^{-\gamma}$ for some constants $\beta,\gamma>0$. Further, assume 
    $\BIR'_1(n)\geq  c\; \BIR_1(n)$,
for some $c>1-\gamma$.
Then for all $n\geq n_0$ and small enough $p>0$ it holds that:
\begin{align*}
 \BIR_2(n) 
    &\geq  (1-p+pc)\; \E[\; \BIR_1(m) \;] \\
\E[\; \BIR_1(m) \;]
    &\geq \BIR_1(\; \E[m] \;) &\qquad\text{(By Jensen's inequality)} \\
    &\geq \BIR_1(\; (1-p)n \;) &\qquad\text{(since $\E[m]\geq n(1-p)$)}  \\
    &\geq \beta\cdot n^{-\gamma} \cdot (1-p)^{-\gamma} 
        &\qquad\text{(plugging in $\BIR_1(n)=\beta n^{-\gamma}$)}  \\
    &> \BIR_1(n)\;\; (1-p\gamma)^{-1}
        &\qquad\text{(since $(1-p)^\gamma < 1-p\gamma$)}.\\
\BIR_2(n) 
    &>\alpha\cdot \BIR_1(n),
    &\text{where} \quad 
    \alpha = \tfrac{1-p+pc}{1-p\gamma}>1.
\end{align*}
(In the above equations, all expectations are over $m$ distributed as in \eqref{eq:random-BIR}.)
\end{remark}

\begin{proof}[Proof of Theorem~\ref{thm:random-greedy}]
  Let $\rew'_1(n)$ denote the Bayesian-expected reward of the ``greedy
  choice'' after after $n-1$ steps of \alg[1]. Note that
  $\rew_1(\cdot)$ and $\rew'_1(\cdot)$ are non-decreasing: the former
  because \alg[1] is monotone and the latter because the ``greedy
  choice'' is only improved with an increasing set of
  observations. Therefore, the modified algorithm \alg[2] is monotone
  by \eqref{eq:random-BIR}.

  By definition of the ``greedy choice,'' $\rew_1(n)\leq \rew'_1(n)$
  for all steps $n$. Moreover, by Lemma~\ref{lm:DG-rew},
  \alg[1] has a strictly smaller $\rew(n_0)$ compared to \DynGreedy;
  so, $\rew_1(n_0)<\rew_2(n_0)$.\sw{fixed inequalities; but this paragraph is still not perfectly clear}

Let $\alg$ denote a copy of \alg[1] that is running ``inside" the modified algorithm \alg[2]. Let $m_2(t)$ be the number of global rounds before $t$ in which the agent chooses principal $2$ \emph{and} \alg is invoked; in other words, it is the number of agents seen by \alg before global round $t$. Let $\mM_{2,t}$ be the agents' posterior distribution for $m_2(t)$.

We claim that in each global round $t\geq n_0$, distribution $\mM_{2,t}$ stochastically dominates distribution $\posteriorN{1}{t}$, and $\PMR_1(t)<\PMR_2(t)$. We use induction on $t$. The base case $t=n_0$ holds because $\mM_{2,t} = \posteriorN{1}{t}$ (because the two algorithms coincide on the first $n_0-1$ steps), and $\PMR_1(n_0)<\PMR_2(n_0)$ is proved as in \eqref{eq:pf:thm:DG-dominance}, using the fact that $\rew_1(n_0)<\rew_2(n_0)$.

The induction step is proved as follows. The induction hypothesis for global round $t-1$ implies that agent $t-1$ is seen by \alg with probability $(1-\eps_0)(1-p)$, which is strictly larger than $\eps_0$, the probability with which this agent is seen by \alg[2]. Therefore, $\mM_{2,t}$ stochastically dominates $\posteriorN{1}{t}$.
\begin{align}
\PMR_1(t)
  &= \Ex{n\sim \posteriorN{1}{t}} {\rew_1(n+1)} \nonumber \\
  &\leq \Ex{m\sim \mM_{2,t}} {\rew_1(m+1)}
    \label{eq:pf:thm:random-greedy-1}\\
  &< \Ex{m\sim \mM_{2,t}} {(1-p)\cdot \rew_1(m+1) + p\cdot \rew'_1(m+1)}
    \label{eq:pf:thm:random-greedy-2} \\
  &= \PMR_2(t). \nonumber
\end{align}
Here inequality \eqref{eq:pf:thm:random-greedy-1} holds because $\rew_1(\cdot)$ is monotone and $\mM_{2,t}$ stochastically dominates $\posteriorN{1}{t}$, and inequality \eqref{eq:pf:thm:random-greedy-2} holds because $\rew_1(n_0)<\rew_2(n_0)$ and $\mM_{2,t}(n_0)>0$.%
\footnote{If $\rew_1(\cdot)$ is strictly increasing, then inequality \eqref{eq:pf:thm:random-greedy-1} is strict, too; this is because $\mM_{2,t}(t-1)>\posteriorN{1}{t}(t-1)$.  }
\end{proof}


\section{SoftMax response function}
\label{sec:soft}
This section is devoted to the \SoftMaxRandom model. We recover a positive result under the assumptions from Theorem~\ref{thm:random-clean} (albeit with a weaker conclusion), and then proceed to a much more challenging result under weaker assumptions. We start with a formal definition:

\begin{definition}\label{def:SoftMax}
A response function $\respF$ is \SoftMaxRandom if the following conditions hold:
\begin{OneLiners}
\item  $\respF(\cdot)$ is bounded away from $0$ and $1$:
    $\respF(\cdot)\in [\eps,1-\eps]$ for some $\eps\in (0,\tfrac12)$,
\item  the response function
 $\respF(\cdot)$ is ``smooth" around $0$:
 \begin{align}\label{eq:SoftMax-smooth}
 \exists\, \text{constants $\delta_0,c_0,c'_0>0$}
    \qquad \forall x\in [-\delta_0,\delta_0] \qquad
    c_0 \leq \respF'(x) \leq c'_0.
 \end{align}
\item fair tie-breaking: $\respF(0)=\tfrac12$.
\end{OneLiners}
\end{definition}

\begin{remark}
\asedit{This definition is fruitful when parameters $c_0$ and $c_0'$ are close to $\tfrac12$. Throughout, we assume that \alg[1] is better than \alg[2], and obtain results parameterized by $c_0$. By symmetry, one could assume that \alg[2] is better than \alg[1], and obtain similar results parameterized by $c_0'$.}
\end{remark}

Our first result is a version of Theorem~\ref{thm:random-clean}, with the same assumptions about the algorithms and essentially the same proof. The conclusion is much weaker: we can only guarantee that each agent $t\geq n_0$ chooses principal 1 with probability slightly larger than $\tfrac12$. This is essentially unavoidable in a typical case when both algorithms satisfy $\BIR(n)\to 0$, by Definition~\ref{def:SoftMax}.

\begin{theorem}\label{thm:SoftMax-weak}
  Assume \SoftMaxRandom response function. Suppose \alg[1] has better
  \BIR in the sense of \eqref{eq:random-better-weaker}, and \alg[2]
  satisfies the condition \eqref{eq:random-assn}. Then each agent
  $t\geq n_0$ chooses principal $1$ with probability
\begin{align}\label{eq:thm:SoftMax-weak}
     \Pr[i_t = 1]\geq \tfrac12 +  \tfrac{c_0}{4}\; \BIR_2(t).
\end{align}
\end{theorem}

\begin{proof}[Proof Sketch]
We follow the steps in the proof of Theorem~\ref{thm:random-clean} to derive \begin{align*}
\PMR_1(t) - \PMR_2(t)
    &\geq \BIR_2(t)/2 -q,
    \quad \text{where $q = \exp(-\eps_0 t/12)$.}
\end{align*}
This is at least $\BIR_2(t)/4$ by \eqref{eq:random-assn}. Then \eqref{eq:thm:SoftMax-weak} follows by the smoothness condition \eqref{eq:SoftMax-smooth}.
\end{proof}

\newcommand{\BReg}{\term{BReg}}

We recover a version of Corollary~\ref{cor:random}, if each principal's utility is the number of users (rather than the more general model in \eqref{eq:general-utility}). We also need a mild technical assumption that cumulative Bayesian regret (\BReg) tends to infinity. \BReg is a standard notion from the literature (along with \BIR):
\begin{align}\label{eq:SoftMax-BReg}
\BReg(n) := n\cdot \E_{\mu\sim\priorMu}
    \left[ \max_{a\in A} \mu_a\right] - \sum_{n=1}^n \rew(n')
    = \sum_{n'=1}^n \BIR(n').
\end{align}

\begin{corollary}\label{cor:SoftMax}
Assume that the response function is \SoftMaxRandom, and each principal's  utility is the number of users.
%
%
Consider the restricted competition game with special algorithm \alg, and assume that all other allowed algorithms satisfy $\BReg(n)\to \infty$. Then, for any sufficiently large time horizon $T$, this game has a unique Nash equilibrium: both principals choose \alg.
\end{corollary}

Further, we prove a much more challenging result in which the
condition \eqref{eq:random-better-weaker} is replaced with a much
weaker ``BIR-dominance'' condition. For clarity, we will again assume
that both algorithms are well-defined for an infinite time
horizon. The \emph{weak BIR dominance} condition says there exist
constants $\beta_0, \alpha_0\in (0, 1/2)$ and $n_0$ such that
 \begin{align}\label{eq:SoftMax-better}
   (\forall n\geq n_0) \quad
   \frac{\BIR_1((1-\beta_0)\, n)}{\BIR_2(n)} <1-\alpha_0.
 \end{align}
 If this holds, we say that \alg[1] \emph{weakly BIR-dominates}
 \alg[2]. Note that the condition \eqref{eq:random-better-messy}
 involves sufficiently small multiplicative factors (resp., $\eps_0/2$
 and $\tfrac12$), the new condition replaces them with factors that
 can be arbitrarily close to $1$.

 We make a mild assumption on \alg[1] that its $\BIR_1(n)$ tends to
 0. Formally, for any $\eps > 0$, there exists some $n(\eps)$ such
 that
 \begin{align}\label{eq:so-mild}
   (\forall n\geq n(\eps)) \qquad \BIR_1(n) \leq \eps.
 \end{align}
 We also require a slightly stronger version of the technical
 assumption~\eqref{eq:random-assn}:
for some $n_0$,
\begin{align}\label{eq:SoftMax-assn-strong}
(\forall n\geq n_0) \qquad 
\BIR_2(n) \geq \frac{4}{\alpha_0} \exp \left( \frac{-\min\{\eps_0, 1/8\} n}{12}\right)
\end{align}

\begin{theorem}\label{thm:SoftMax-strong}
  Assume the \SoftMaxRandom response function. Suppose \alg[1]
  weakly-BIR-dominates \alg[2], \alg[1] satisfies \eqref{eq:so-mild},
  and \alg[2] satisfies \eqref{eq:SoftMax-assn-strong}. Then there
  exists some $t_0$ such that each agent $t\geq t_0$ chooses principal
  $1$ with probability
\begin{align}\label{eq:thm:SoftMax-strong}
     \Pr[i_t = 1]\geq \tfrac12 +  \tfrac{c_0\alpha_0}{4}\; \BIR_2(t).
\end{align}
\end{theorem}

The main idea behind our proof is that even though \alg[1] may have a
slower rate of learning in the beginning, it will gradually catch up
and surpass \alg[2]. We will describe this process in two phases. In
the first phase, \alg[1] receives a random agent with probability at
least $\respF(-1) = \eps_0$ in each round. Since $\BIR_1$ tends to 0,
the difference in \BIR{s} between the two algorithms is also
diminishing. Due to the \SoftMaxRandom response function, \alg[1]
attracts each agent with probability at least $1/2 - O(\beta_0)$ after
a sufficient number of rounds. Then the game enters the second phase:
both algorithms receive agents at a rate close to $\tfrac12$, and the
fractions of agents received by both algorithms --- $n_1(t)/t$ and
$n_2(t)/t$ --- also converge to $\tfrac12$. At the end of the second
phase and in each global round afterwards, the counts $n_1(t)$ and
$n_2(t)$ satisfy the weak BIR-dominance condition, in the sense that
they both are larger than $n_0$ and $n_1(t)\geq (1-\beta_0)\; n_2(t)$.
At this point, \alg[1] actually has smaller $\BIR$, which reflected in the {\PMR}s eventually. Accordingly, from then on \alg[1]
attracts agents at a rate slightly larger than $\tfrac12$. We prove
that the ``bump'' over $\tfrac12$ is at least on the order of
$\BIR_2(t)$.

\begin{proof}[Proof of Theorem~\ref{thm:SoftMax-strong}]
  Let $\beta_1 = \min\{c_0'\delta_0, \beta_0/20\}$ with $\delta_0$
  defined in~\eqref{eq:SoftMax-smooth}.  Recall each agent chooses
  \alg[1] with probability at least $\respF(-1)= \eps_0$.  By
  condition \eqref{eq:so-mild} and \eqref{eq:SoftMax-assn-strong},
  there exists some sufficiently large $T_1$ such that for any
  $t\geq T_1$, $\BIR_1(\eps_0 T_1/2) \leq \beta_1/c_0'$ and
  $\BIR_2(t) > e^{-\eps_0 t/12}$. Moreover, for any $t\geq T_1$, we
  know $\E[n_1(t+1)] \geq \eps_0\,t $, and by the Chernoff Bounds
  (Theorem~\ref{thm:chernoff}), we have $n_1(t+1) \geq \eps_0 t/2$
  holds with probability at least $1 - q_1(t)$ with
  $q_1(t) = \exp(-\eps_0 t/12) < \BIR_2(t)$. It follows that for any $t\geq T_1$,
\begin{align*}
  \PMR_2(t) - \PMR_1(t) &= \Ex{m_1\sim \posteriorN{1}{t},\;m_2\sim \posteriorN{2}{t}}{\BIR_1(m_1+ 1) - \BIR_2(m_2+1)} \\
                        &\leq q_1(t)  + \Ex{m_1\sim \posteriorN{1}{t}}{\BIR_1(m_1+ 1)\mid m_1 \geq \eps_0 t/2 - 1 } - \BIR_2(t)\\
                        &\leq \BIR_1(\eps_0 T_1/2) \leq \beta_1/c_0'
\end{align*}
Since the response function $\respF$ is $c_0'$-Lipschitz in the
neighborhood of $[-\delta_0, \delta_0]$, each agent after round $T_1$
will choose \alg[1] with probability at least
\[
  p_t \geq \frac{1}{2} - c_0'\left(\PMR_2(t) - \PMR_1(t)\right) \geq
  \frac{1}{2} - \beta_1.
\]

Next, we will show that there exists a sufficiently large $T_2$ such
that for any $t\geq T_1 + T_2$, with high probability
$n_1(t) > \max\{n_0, (1 - \beta_0)n_2(t)\}$, where $n_0$ is defined
in~\eqref{eq:SoftMax-better}. 
Fix any $t \geq T_1 + T_2$. 
Since each agent chooses \alg[1] with probability at least
$1/2 - \beta_1$, by Chernoff Bounds (Theorem~\ref{thm:chernoff}) we
have with probability at least $1 - q_2(t)$ that the number of agents
that choose \alg[1] is at least $\beta_0(1/2 - \beta_1)t/5$, where the
function
$$
q_2(x) = \exp\left( \frac{-(1/2 - \beta_1)(1 - \beta_0/5)^2x}{3} \right).
$$
Note that the number of agents received by \alg[2] is at most
$T_1 + (1/2 + \beta_1)t + (1/2 - \beta_1)(1 - \beta_0/5)t$.

Then as long as
$T_2 \geq \frac{5T_1}{\beta_0}$, we can guarantee that
$n_1(t) > n_2(t) (1 - \beta_0)$ and $n_1(t) > n_0$ with probability at
least $1 - q_2(t)$ for any $t \geq T_1 + T_2$.
Note that the weak BIR-dominance condition
in~\eqref{eq:SoftMax-better} implies that for any $t\geq T_1 + T_2$
with probability at least $1 - q_2(t)$,
\[
  \BIR_1(n_1(t)) < (1- \alpha_0)\BIR_2(n_2(t)).
\]
It follows that for any $t\geq T_1 + T_2$,
\begin{align*}
  \PMR_1(t) - \PMR_2(t) &= \Ex{m_1\sim \posteriorN{1}{t},\;m_2\sim \posteriorN{2}{t}}{\BIR_2(m_2+ 1) - \BIR_1(m_1+1)} \\
                        &\geq (1 - q_2(t))\alpha_0 \BIR_2(t) - q_2(t)\\
                        &\geq \alpha_0 \BIR_2(t)/4
\end{align*}
where the last inequality holds as long as
$q_2(t) \leq \alpha_0\BIR_2(t)/4$, and is implied by the condition
in~\eqref{eq:SoftMax-assn-strong} as long as $T_2$ is sufficiently
large. Hence, by the definition of our \SoftMaxRandom response
function and assumption in~\eqref{eq:SoftMax-smooth}, we have
\[
  \Pr[i_t = 1] \geq \frac{1}{2} + \frac{c_0\alpha_0\BIR_2(t)}{4}. \qedhere
\]
\end{proof}

Similar to the condition \eqref{eq:random-better-weaker}, we can also
generalize the weak BIR-dominance condition \eqref{eq:SoftMax-better}
to handle the dependence on $T$: there exist some $T_0$, a function
$n_0(T)\in \polylog(T)$, and constants $\beta_0,\alpha_0\in (0, 1/2)$, such that 
\begin{align}\label{eq:SoftMax-better-weaker}
(\forall T\geq T_0,  n\geq n_0(T)) \quad
\frac{\BIR_1((1-\beta_0)\, n\mid T)}{\BIR_2(n\mid T)} <1-\alpha_0.
\end{align}

We also provide a version of Theorem~\ref{thm:SoftMax-weak} under this
more general weak BIR-dominance condition; its proof is very similar
and is omitted. The following is just a direct consequence of
Theorem~\ref{thm:SoftMax-weak} with this general condition.\sw{added}

\begin{corollary}\label{cor:SoftMax-strong}
Assume that the response function is \SoftMaxRandom, and each principal's  utility is the number of users. Consider the restricted competition game in which the special algorithm \alg weakly-BIR-dominates the other allowed algorithms, and the latter satisfy $\BReg(n)\to \infty$. Then, for any sufficiently large time horizon $T$, there is a unique Nash equilibrium: both principals choose \alg.
\end{corollary}



\section{Economic implications}
\label{sec:welfare}
\asedit{We frame our contributions in terms of the relationship between \competitiveness and \rationality on one side, and adoption of better algorithms on the other. Recall that both \competitiveness (of the game between the two principals) and \rationality (of the agents) are controlled by the response function $\respF$.}

\OMIT{ 
We frame our contributions in terms of the relationship between \competition and \innovation, \ie between the extent to which the game between the two principals is competitive, and the degree of innovation --- adoption of better that these models incentivize. \Competition is controlled via the response function $\respF$, and \innovation refers to the quality of the technology (MAB algorithms) adopted by the principals. The \competition vs. \innovation relationship is well-studied in the economics literature, and is commonly known to often follow an inverted-U shape, as in \reffig{fig:inverted-U} (see Section~\ref{sec:related-work} for citations). \Competition in our models is closely correlated with \rationality: the extent to which agents make rational decisions, and indeed \rationality is what $\respF$ controls directly.
} 

\xhdr{Main story.}
Our main story concerns the restricted competition game between the two principals where one allowed algorithm \alg is ``better" than the others. \asedit{We track whether and when \alg is chosen in an equilibrium.} We vary \competitiveness/\rationality by changing the response function from \HardMax (full rationality, very competitive environment) to \HardMaxRandom to  \SoftMaxRandom (less rationality and competition). Our conclusions are as follows:
\begin{OneLiners}
\item Under \HardMax, no innovation: \DynGreedy is chosen over \alg.
\item Under \HardMaxRandom, some innovation:  \alg is chosen as long as it BIR-dominates.
\item Under \SoftMaxRandom, more innovation: \alg is chosen as long as it weakly-BIR-dominates.%
\footnote{This is a weaker condition, the better algorithm is chosen in a broader range of scenarios.}
\end{OneLiners}
These conclusions follow, respectively, from Corollaries~\ref{cor:DG-dominance}, \ref{cor:random} and \ref{cor:SoftMax}. Further, \asedit{we consider the uniform choice between the principals. It corresponds to the least amount of rationality and competition, and (when principals' utility is the number of agents) uniform choice provides no incentives to innovate.}%
\footnote{On the other hand, if principals' utility is somewhat aligned with agents' welfare, as in \eqref{eq:general-utility}, then a monopolist principal is incentivized to choose the best possible MAB algorithm (namely, to minimize cumulative Bayesian regret $\BReg(T)$). Accordingly, monopoly would result in better social welfare than competition, as the latter is likely to split the market and cause each principal to learn more slowly. This is a very generic and well-known effect regarding economies of scale.}
Thus, we have an inverted-U relationship, see \reffig{fig:inverted-U2}.

\begin{figure}
\begin{center}
\begin{tikzpicture}[scale=1]
      \draw[->] (-.5,0) -- (9.5,0) node[above] 
        {\qquad\qquad Competitiveness/Rationality};
      \draw[->] (0,-.5) -- (0,3) node[above] {Better algorithm in equilibrium};
      \draw[scale=0.8,domain=0.5:9.5,smooth,variable=\x,blue, line width=0.3mm] plot ({\x},{3.5 - 0.15*(\x - 5)^2});
     \node[below] at (1, 0) {\footnotesize \Uniform};
     \node[below] at (3.9, 0) {\footnotesize \SoftMaxRandom};
     \node[below] at (6, 0) {\footnotesize \HardMaxRandom};
     \node[below] at (8, 0) {\footnotesize \HardMax};
 \end{tikzpicture}

\caption{The stylized inverted-U relationship in the ``main story".}
\label{fig:inverted-U2}
\end{center}
\end{figure}
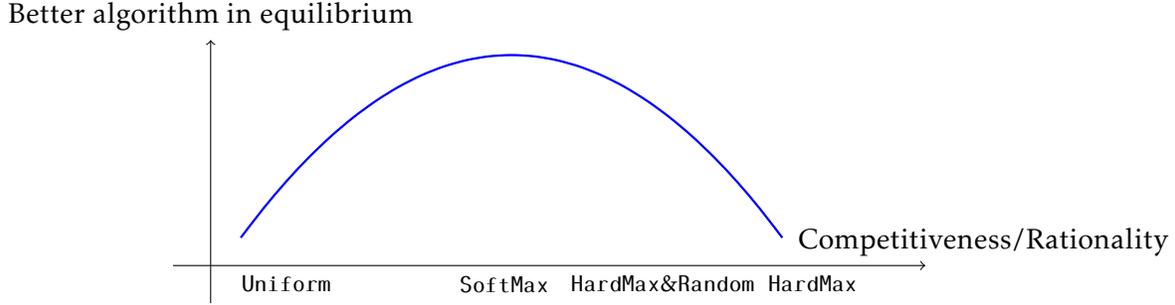

\xhdr{Secondary story.}
Let us zoom in on the symmetric  \HardMaxRandom model. \asedit{Competitiveness and rationality within this model are controlled by the baseline probability $\eps_0 = \respF(- 1)$, which goes smoothly between the two extremes of \HardMax ($\eps_0=0$) and the uniform choice ($\eps_0=\tfrac12$). Smaller $\eps_0$ corresponds to increased rationality and increased competitiveness.} For clarity, we assume that principal's utility is the number of agents.

We consider the marginal utility of switching to a better algorithm. Suppose initially both principals use some algorithm \alg, and principal 1 ponders switching to another algorithm \alg' which BIR-dominates \alg. \asedit{We are interested in the marginal utility of this switch. Then:}

\begin{itemize}
\item $\eps_0 = 0$ (\HardMax):~~~~ the marginal utility can be negative if \alg is \DynGreedy.

\item $\eps_0$ near $0$:~~~~ only a small marginal utility can be guaranteed, as it may take a long time for $\alg'$ to ``catch up" with \alg, and hence less time to reap the benefits.

\item ``medium-range" $\eps_0$:~~~~ large marginal utility, as $\alg'$ learns fast and gets most agents.

\item $\eps_0$ near $\tfrac12$:~~~~ small marginal utility, as principal 1 gets most agents for free no matter what.
\end{itemize}
The familiar inverted-U shape is depicted in Figure~\ref{fig:inverted-U3}.

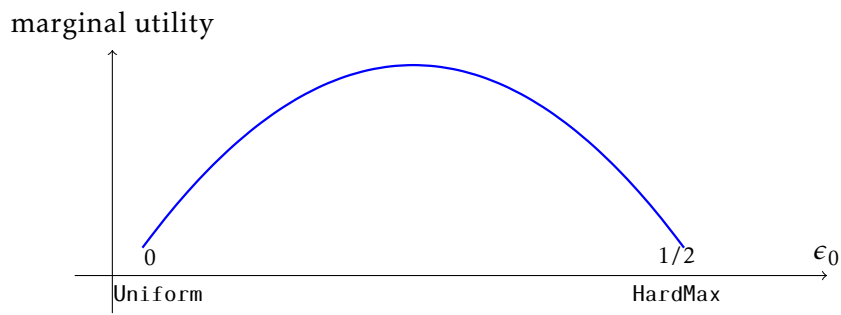
\begin{figure}
\begin{center}
\begin{tikzpicture}[scale=1]
      \draw[->] (-.5,0) -- (9.5,0) node[above]  {$\eps_0$};
      \draw[->] (0,-.5) -- (0,3) node[above] {marginal utility};
      \draw[scale=0.8,domain=0.5:9.5,smooth,variable=\x,blue, line width=0.3mm] plot ({\x},{3.5 - 0.15*(\x - 5)^2});
     \node[below] at (.6, 0) {\footnotesize \Uniform};
     \node[above] at (.5, 0) {\footnotesize 0};
     \node[below] at (7.5,0) {\footnotesize \HardMax};
     \node[above] at (7.5, 0) {\footnotesize 1/2};
 \end{tikzpicture}

\caption{The stylized inverted-U relationship from the ``secondary story"}
\label{fig:inverted-U3}
\end{center}
\end{figure}


\section*{Acknowledgements}
The authors would like to thank Glen Weyl for discussions of related work in economics.

\bibliographystyle{plainnat}
\begin{small}
\bibliography{bib-abbrv,bib-AGT,bib-bandits,bib-slivkins,bib-random,bib-ML}

\begin{thebibliography}{44}
\providecommand{\natexlab}[1]{#1}
\providecommand{\url}[1]{\texttt{#1}}
\expandafter\ifx\csname urlstyle\endcsname\relax
  \providecommand{\doi}[1]{doi: #1}\else
  \providecommand{\doi}{doi: \begingroup \urlstyle{rm}\Url}\fi

\bibitem[Agarwal et~al.(2016)Agarwal, Bird, Cozowicz, Dudik, Langford, Li,
  Hoang, Melamed, Sen, Schapire, and Slivkins]{MWT-WhitePaper-2016}
Alekh Agarwal, Sarah Bird, Markus Cozowicz, Miro Dudik, John Langford, Lihong
  Li, Luong Hoang, Dan Melamed, Siddhartha Sen, Robert Schapire, and Alex
  Slivkins.
\newblock Multiworld testing: A system for experimentation, learning, and
  decision-making, 2016.
\newblock A white paper, available at
  \url{https://github.com/Microsoft/mwt-ds/raw/master/images/MWT-WhitePaper.pdf}.

\bibitem[Aghion et~al.(2005)Aghion, Bloom, Blundell, Griffith, and
  Howitt]{Aghion-QJE05}
Philippe Aghion, Nicholas Bloom, Richard Blundell, Rachel Griffith, and Peter
  Howitt.
\newblock Competition and innovation: An inverted u relationship.
\newblock \emph{Quaterly J. of Economics}, 120\penalty0 (2):\penalty0 701--728,
  2005.

\bibitem[Athey and Segal(2013)]{AtheySegal-econometrica13}
Susan Athey and Ilya Segal.
\newblock An efficient dynamic mechanism.
\newblock \emph{Econometrica}, 81\penalty0 (6):\penalty0 2463--2485, November
  2013.
\newblock A preliminary version has been available as a working paper since
  2007.

\bibitem[Auer et~al.(2002{\natexlab{a}})Auer, Cesa-Bianchi, and
  Fischer]{bandits-ucb1}
Peter Auer, Nicol{\`o} Cesa-Bianchi, and Paul Fischer.
\newblock Finite-time analysis of the multiarmed bandit problem.
\newblock \emph{Machine Learning}, 47\penalty0 (2-3):\penalty0 235--256,
  2002{\natexlab{a}}.

\bibitem[Auer et~al.(2002{\natexlab{b}})Auer, Cesa-Bianchi, Freund, and
  Schapire]{bandits-exp3}
Peter Auer, Nicol{\`o} Cesa-Bianchi, Yoav Freund, and Robert~E. Schapire.
\newblock The nonstochastic multiarmed bandit problem.
\newblock \emph{SIAM J. Comput.}, 32\penalty0 (1):\penalty0 48--77,
  2002{\natexlab{b}}.
\newblock Preliminary version in {\em 36th IEEE FOCS}, 1995.

\bibitem[Azevedo and Gottlieb(2017)]{Azevedo-Gottlieb-17}
Eduardo Azevedo and Daniel Gottlieb.
\newblock Perfect competition in markets with adverse selection.
\newblock \emph{Econometrica}, 85\penalty0 (1):\penalty0 67--105, 2017.

\bibitem[Babaioff et~al.(2014)Babaioff, Sharma, and Slivkins]{MechMAB-ec09}
Moshe Babaioff, Yogeshwer Sharma, and Aleksandrs Slivkins.
\newblock Characterizing truthful multi-armed bandit mechanisms.
\newblock \emph{SIAM J. on Computing (SICOMP)}, 43\penalty0 (1):\penalty0
  194--230, 2014.
\newblock Preliminary version in \emph{10th ACM EC}, 2009.

\bibitem[Babaioff et~al.(2015)Babaioff, Kleinberg, and
  Slivkins]{Transform-ec10-jacm}
Moshe Babaioff, Robert Kleinberg, and Aleksandrs Slivkins.
\newblock Truthful mechanisms with implicit payment computation.
\newblock \emph{J. of the ACM}, 62\penalty0 (2):\penalty0 10, 2015.
\newblock Subsumes the conference papers in \emph{ACM EC 2010} and \emph{ACM EC
  2013}.

\bibitem[Bahar et~al.(2016)Bahar, Smorodinsky, and Tennenholtz]{Bahar-ec16}
Gal Bahar, Rann Smorodinsky, and Moshe Tennenholtz.
\newblock Economic recommendation systems.
\newblock In \emph{16th ACM Conf. on Electronic Commerce (EC)}, 2016.

\bibitem[Bergemann and V\"{a}lim\"{a}ki(2010)]{DynPivot-econometrica10}
Dirk Bergemann and Juuso V\"{a}lim\"{a}ki.
\newblock The dynamic pivot mechanism.
\newblock \emph{Econometrica}, 78\penalty0 (2):\penalty0 771--789, 2010.
\newblock Preliminary versions have been available since 2006, as \emph{Cowles
  Foundation Discussion Papers} \#1584 (2006), \#1616 (2007) and \#1672(2008).

\bibitem[Bimpikis et~al.(2017)Bimpikis, Papanastasiou, and
  Savva]{Bimpikis-exploration-ms17}
Kostas Bimpikis, Yiangos Papanastasiou, and Nicos Savva.
\newblock Crowdsourcing exploration.
\newblock \emph{Management Science}, 2017.
\newblock Forthcoming.

\bibitem[Bolton and Harris(1999)]{Bolton-econometrica99}
Patrick Bolton and Christopher Harris.
\newblock {Strategic Experimentation}.
\newblock \emph{Econometrica}, 67\penalty0 (2):\penalty0 349--374, 1999.

\bibitem[Bubeck and Cesa-Bianchi(2012)]{Bubeck-survey12}
S\'{e}bastien Bubeck and Nicolo Cesa-Bianchi.
\newblock {Regret Analysis of Stochastic and Nonstochastic Multi-armed Bandit
  Problems}.
\newblock \emph{Foundations and Trends in Machine Learning}, 5\penalty0 (1),
  2012.

\bibitem[Che and H\"{o}rner(2015)]{Che-13}
Yeon-Koo Che and Johannes H\"{o}rner.
\newblock Optimal design for social learning.
\newblock Preprint, 2015.
\newblock First draft: 2013.

\bibitem[Devanur and Kakade(2009)]{DevanurK09}
Nikhil Devanur and Sham~M. Kakade.
\newblock The price of truthfulness for pay-per-click auctions.
\newblock In \emph{10th ACM Conf. on Electronic Commerce (EC)}, pages 99--106,
  2009.

\bibitem[Even{-}Dar et~al.(2006)Even{-}Dar, Mannor, and
  Mansour]{EvenDar-icml06}
Eyal Even{-}Dar, Shie Mannor, and Yishay Mansour.
\newblock Action elimination and stopping conditions for the multi-armed bandit
  and reinforcement learning problems.
\newblock \emph{J. of Machine Learning Research (JMLR)}, 7:\penalty0
  1079--1105, 2006.

\bibitem[Frazier et~al.(2014)Frazier, Kempe, Kleinberg, and
  Kleinberg]{Frazier-ec14}
Peter Frazier, David Kempe, Jon~M. Kleinberg, and Robert Kleinberg.
\newblock Incentivizing exploration.
\newblock In \emph{ACM Conf. on Economics and Computation (ACM EC)}, pages
  5--22, 2014.

\bibitem[Gabaix et~al.(2016)Gabaix, Laibson, Li, Li, Resnick, and
  de~Vries]{Gabaix-16}
Xavier Gabaix, David Laibson, Deyuan Li, Hongyi Li, Sidney Resnick, and
  Casper~G. de~Vries.
\newblock The impact of competition on prices with numerous firms.
\newblock \emph{J. of Economic Theory}, 165:\penalty0 1--24, 2016.

\bibitem[Ghosh and Hummel(2013)]{Ghosh-itcs13}
Arpita Ghosh and Patrick Hummel.
\newblock Learning and incentives in user-generated content: multi-armed
  bandits with endogenous arms.
\newblock In \emph{Innovations in Theoretical Computer Science Conf. (ITCS)},
  pages 233--246, 2013.

\bibitem[Gittins et~al.(2011)Gittins, Glazebrook, and Weber]{Gittins-book11}
John Gittins, Kevin Glazebrook, and Richard Weber.
\newblock \emph{{Multi-Armed Bandit Allocation Indices}}.
\newblock John Wiley \& Sons, 2011.

\bibitem[Gummadi et~al.(2012)Gummadi, Johari, and Yu]{Johari-ec12}
Ramakrishna Gummadi, Ramesh Johari, and Jia~Yuan Yu.
\newblock Mean field equilibria of multiarmed bandit games.
\newblock In \emph{13th ACM Conf. on Electronic Commerce (EC)}, 2012.

\bibitem[Ho et~al.(2016)Ho, Slivkins, and Vaughan]{RepeatedPA-ec14}
Chien-Ju Ho, Aleksandrs Slivkins, and Jennifer~Wortman Vaughan.
\newblock Adaptive contract design for crowdsourcing markets: Bandit algorithms
  for repeated principal-agent problems.
\newblock \emph{J. of Artificial Intelligence Research}, 55:\penalty0 317--359,
  2016.
\newblock Preliminary version appeared in \emph{ACM EC 2014}.

\bibitem[Hotelling(1929)]{Hotelling-29}
Harold Hotelling.
\newblock Stability in competition.
\newblock \emph{The Economic Journal}, 39\penalty0 (153):\penalty0 41--57,
  1929.

\bibitem[Immorlica et~al.(2011)Immorlica, Kalai, Lucier, Moitra, Postlewaite,
  and Tennenholtz]{DuelingAlgs-stoc11}
Nicole Immorlica, Adam~Tauman Kalai, Brendan Lucier, Ankur Moitra, Andrew
  Postlewaite, and Moshe Tennenholtz.
\newblock Dueling algorithms.
\newblock In \emph{43rd ACM Symp. on Theory of Computing (STOC)}, pages
  215--224, 2011.

\bibitem[Kakade et~al.(2013)Kakade, Lobel, and Nazerzadeh]{Kakade-pivot-or13}
Sham~M. Kakade, Ilan Lobel, and Hamid Nazerzadeh.
\newblock Optimal dynamic mechanism design and the virtual-pivot mechanism.
\newblock \emph{Operations Research}, 61\penalty0 (4):\penalty0 837--854, 2013.

\bibitem[Keller et~al.(2005)Keller, Rady, and Cripps]{Keller-econometrica05}
Godfrey Keller, Sven Rady, and Martin Cripps.
\newblock {Strategic Experimentation with Exponential Bandits}.
\newblock \emph{Econometrica}, 73\penalty0 (1):\penalty0 39--68, 2005.

\bibitem[Kleinberg et~al.(2016)Kleinberg, Waggoner, and Weyl]{Bobby-Glen-ec16}
Robert~D. Kleinberg, Bo~Waggoner, and E.~Glen Weyl.
\newblock Descending price optimally coordinates search.
\newblock Working paper, 2016.
\newblock Preliminary version in \emph{ACM EC 2016}. Under submission to
  \emph{Econometrica}.

\bibitem[Kremer et~al.(2014)Kremer, Mansour, and Perry]{Kremer-JPE14}
Ilan Kremer, Yishay Mansour, and Motty Perry.
\newblock Implementing the “wisdom of the crowd”.
\newblock \emph{J. of Political Economy}, 122:\penalty0 988--1012, 2014.
\newblock Preliminary version in \emph{ACM EC 2014}.

\bibitem[Lai and Robbins(1985)]{Lai-Robbins-85}
Tze~Leung Lai and Herbert Robbins.
\newblock {Asymptotically efficient Adaptive Allocation Rules}.
\newblock \emph{Advances in Applied Mathematics}, 6:\penalty0 4--22, 1985.

\bibitem[Mansour et~al.(2015)Mansour, Slivkins, and
  Syrgkanis]{ICexploration-ec15}
Yishay Mansour, Aleksandrs Slivkins, and Vasilis Syrgkanis.
\newblock Bayesian incentive-compatible bandit exploration.
\newblock In \emph{15th ACM Conf. on Electronic Commerce (EC)}, 2015.

\bibitem[Mansour et~al.(2016)Mansour, Slivkins, Syrgkanis, and
  Wu]{ICexplorationGames-ec16-working}
Yishay Mansour, Aleksandrs Slivkins, Vasilis Syrgkanis, and Steven Wu.
\newblock Bayesian exploration: Incentivizing exploration in bayesian games.
\newblock Working paper, 2016.
\newblock available at {\tt https://arxiv.org/abs/1602.07570}. Preliminary
  version in \emph{ACM EC 2016}.

\bibitem[Milgrom and Stokey(1982)]{Milgrom-Stokey-82}
Paul Milgrom and Nancy Stokey.
\newblock Information, trade and common knowledge.
\newblock \emph{J. of Economic Theory}, 26\penalty0 (1):\penalty0 17--27, 1982.

\bibitem[Mitzenmacher and Upfal(2005)]{MitzUpfal-book05}
Michael Mitzenmacher and Eli Upfal.
\newblock \emph{{Probability and Computing: Randomized Algorithms and
  Probabilistic Analysis}}.
\newblock Cambridge University Press, 2005.

\bibitem[Perloff and Salop(1985)]{Perloff-Salop-85}
Jeffrey~M. Perloff and Steven~C. Salop.
\newblock Equilibrium with product differentiation.
\newblock \emph{Review of Economic Studies}, LII:\penalty0 107--120, 1985.

\bibitem[Rothschild and Stiglitz(1976)]{Rothschild-Stiglitz-76}
Michael Rothschild and Joseph Stiglitz.
\newblock Equilibrium in competitive insurance markets: An essay on the
  economics of imperfect information.
\newblock \emph{Quaterly J. of Economics}, 90\penalty0 (4):\penalty0 629--649,
  1976.

\bibitem[Rysman(2009)]{Rysman09}
Marc Rysman.
\newblock The economics of two-sided markets.
\newblock \emph{J. of Economic Perspectives}, 23\penalty0 (3):\penalty0
  125--144, 2009.

\bibitem[Schumpeter(1942)]{Schumpeter-42}
Joseph Schumpeter.
\newblock \emph{Capitalism, Socialism and Democracy}.
\newblock Harper \& Brothers, 1942.

\bibitem[Singla and Krause(2013)]{Krause-www13}
Adish Singla and Andreas Krause.
\newblock Truthful incentives in crowdsourcing tasks using regret minimization
  mechanisms.
\newblock In \emph{22nd Intl. World Wide Web Conf. (WWW)}, pages 1167--1178,
  2013.

\bibitem[Slivkins(2017)]{slivkins-MABbook}
Aleksandrs Slivkins.
\newblock Introduction to multi-armed bandits, 2017.
\newblock A book draft, available at
  {http://research.microsoft.com/en-us/people/slivkins}.

\bibitem[Veiga and Weyl(2016)]{Veiga-Weyl-16}
Andre Veiga and Glen Weyl.
\newblock Product design in selection markets.
\newblock \emph{Quarterly J. of Economics}, 131\penalty0 (2):\penalty0
  1007--1056, 2016.

\bibitem[Vives(2008)]{Vives-08}
Xavier Vives.
\newblock Innovation and competitive pressure.
\newblock \emph{J. of Industrial Economics}, 56\penalty0 (3), 2008.

\bibitem[Weyl and White(2014)]{Weyl-White-14}
Glen Weyl and Alexander White.
\newblock Let the right ‘one' win: Policy lessons from the new economics of
  platforms.
\newblock \emph{Competition Policy International}, 12\penalty0 (2):\penalty0
  29--51, 2014.

\bibitem[Yue and Joachims(2009)]{Yue-dueling-icml09}
Yisong Yue and Thorsten Joachims.
\newblock Interactively optimizing information retrieval systems as a dueling
  bandits problem.
\newblock In \emph{26th Intl. Conf. on Machine Learning (ICML)}, pages
  1201--1208, 2009.

\bibitem[Yue et~al.(2012)Yue, Broder, Kleinberg, and Joachims]{Yue-dueling12}
Yisong Yue, Josef Broder, Robert Kleinberg, and Thorsten Joachims.
\newblock The k-armed dueling bandits problem.
\newblock \emph{J. Comput. Syst. Sci.}, 78\penalty0 (5):\penalty0 1538--1556,
  2012.
\newblock Preliminary version in COLT 2009.

\end{thebibliography}
\end{small}

\appendix

\section{Background on multi-armed bandits}
\label{app:examples}

\newcommand{\ExplorExploit}{\term{ExplorExploit}}
\newcommand{\PhasedExplorExploit}{\term{PhasedExplorExploit}}
\newcommand{\RandomDynGreedy}{\term{RandomDynGreedy}}
\newcommand{\SuccesiveElimination}{\term{SuccesiveElimination}}
\newcommand{\SuccesiveEliminationReset}{\term{SuccesiveEliminationReset}}

This appendix provides some pertinent background on multi-armed
bandits (\emph{MAB}). We discuss \BIR and monotonicity of several MAB algorithms, touching upon: \DynGreedy and \StaticGreedy (Section~\ref{sec:MAB-greedy}), ``naive" MAB algorithms that separate exploration and exploitation (Section~\ref{sec:MAB-naive}), and ``smart" MAB algorithms that combine exploration and exploitation (Section~\ref{sec:MAB-smart}).

As we do throughout the paper, we focus on MAB with i.i.d. rewards and a Bayesian prior; we call it \emph{Bayesian MAB} for brevity.


\subsection{\DynGreedy and \StaticGreedy}
\label{sec:MAB-greedy}

We provide an example when \DynGreedy and \StaticGreedy have
constant \BIR, and prove monotonicity of \DynGreedy. For the
example, it suffices to consider \emph{deterministic rewards} (for
each action $a$, the realized reward is always equal to the mean
$\mu_a$) and \emph{independent priors} (according to the prior
$\priorMu$, random variables $\mu_1 \LDOTS \mu_K$ are mutually
independent) each of {\em full support}.

%
%
%
%
%
%

The following claim is immediate from the definition of the CDF
function
\begin{claim}
Assume independent priors. Let $F_i$ be the CDF of the mean reward
$\mu_i$ of action $a_i\in A$. Then, for any numbers
$z_2>z_1>\E[\mu_2]$ we have
    $\Pr[\text{$\mu_1\leq z_1$ and $\mu_2\geq z_2$}] = F_1(z_1)(1-F_2(z_2)) $.
\end{claim}



We can now draw an immediate corollary of the above claim

\begin{corollary}
Consider any problem instance of Bayesian MAB with two actions and independent
priors which are full support. Then:
\begin{OneLiners}
\item[(a)] With constant probability, \StaticGreedy  has a constant \BIR for all steps.
\item[(b)] Assuming deterministic rewards, with
constant probability \DynGreedy has a constant \BIR for all steps.
\end{OneLiners}
\end{corollary}

\begin{remark}
A similar result holds for  rewards which are distributed as
Bernoulli random variables. In this case we consider accumulative
reward of an action as a random walk, and use a high probability
variation of the law of iterated logarithms. (Details omitted.)
\end{remark}

Next, we show that \DynGreedy is monotone.

\begin{lemma}\label{dgmono}
\DynGreedy is monotone, in the sense that $\rew(n)$ is non-decreasing.
Further, $\rew(n)$ is strictly increasing for every time step $n$ with $\Pr[a_n\neq a_{n+1}]>0$.
\end{lemma}

\begin{proof}
We prove by induction on $n$ that $\rew(n)\leq \rew(n+1)$ for
\DynGreedy. Let $a_n$ be the random variable recommended at time
$t$, then $\E[\mu_{a_n}| \mI_n ]=\rew(n)$. We can rewrite this as:
\[
\rew(n)=\E_{\mI_n}[\E_{r_n}[\mu_{a_n}|r_n,\mI_n]] =
\E_{\mI_{n+1}}[\mu_{a_n}|\mI_{n+1}]
\]
since $\mI_{n+1}=(\mI_n,r_n)$. At time $n+1$ \DynGreedy will select
an action $a_{n+1}$ such that:
\[
\rew(n+1)=\E[\mu_{a_{n+1}}|\mI_{n+1}]\geq \E[\mu_{a_n}
|\mI_n]=\rew(n)
\]
%
which proves the monotonicity. In cases that $\Pr[a_n\neq
a_{n+1}]>0]$ we have a strict inequality, since with some
probability we select a better action then the realization of $a_n$.
\end{proof}


\subsection{``Naive" MAB algorithms that separate exploration and exploitation}
\label{sec:MAB-naive}


MAB algorithm \ExplorExploit$(m)$ initially explores each action
with $m$ agents and for the remaining $T-|A|m$ agents recommends the
action with the highest observed average. In the explore phase it
assigns a random permutation of the $mK$ recommendations.

\begin{lemma}
The \ExplorExploit$(T^{2/3}\log |A|/\delta)$ algorithm has, with
probability $1-\delta$, for any  $n\geq |A|T^{2/3}$ we have
\BIR$(n)=O(T^{-1/3})$. In addition, \ExplorExploit$(m)$ is monotone.
\end{lemma}

\begin{proof}
In the explore phase we we approximate for each action $a\in A$, the
value of $\mu_a$ by $\hat{\mu}_a$. Using the standard Chernoff
bounds we have that with probability $1-\delta$, for every action
$a\in A$ we have $|\mu_a -\hat{\mu}_a| \leq T^{-1/3}$.

Let $a^* = \arg\max_a \mu_a$ and $a^{ee}$ the action that
\ExplorExploit selects in the explore phase after the first
$|A|T^{2/3}$ agents. Since $\hat{\mu}_{a^*} \leq
\hat{\mu}_{a^{ee}}$, this implies that $\mu_{a^*} -
\mu_{a^{ee}}=O(T^{-1/3})$.

To show that \ExplorExploit$(m)$ is monotone, we need to show only
that $\rew(mK) \leq \rew(mK+1)$. This follows since for any $t< mK$
we have $\rew(t)=\rew(t+1)$, since the recommended action is
uniformly distributed for each time $t$. Also, for any $t\geq mK+1$
we have $\rew(t)=\rew(t+1)$ since we are recommending the same
exploration action. The proof that $\rew(mK) \leq \rew(mK+1)$ is the
same as for \DynGreedy in Lemma~\ref{dgmono}.
\end{proof}

We can also have a a phased version which we call
\PhasedExplorExploit$(m_t)$, where time is partition in to phases.
In phase $t$ we have $m_t$ agents and a random subset of $K$ explore
the actions (each action explored by a single agent) and the other
agents exploit. (This implies that we need that $m_t\geq K$ for all
$t$. We also assume that $m_t$ is monotone in $t$.)

\begin{lemma}
Consider the case that $K=2$ and the rewards of the actions are
Bernoulli r.v. with parameter $\mu_i$ and $\Delta=\mu_1-\mu_2$.
Algorithm \PhasedExplorExploit$(m_t)$ is monotone and for $m_t =
\sqrt{t}$ it has $\BIR(n)=O(n^{-1/3}+e^{-O(\Delta^2 n^{2/3})}))$. 
\end{lemma}

\begin{proof}
We first show that it is monotone. Recall that $\mu_1>\mu_2$. Let
$S_i=\sum_{j=1}^t r_{i,j}$ be the sum of the rewards of action $i$
up to phase $t$. We need to show that $\Pr[S_1>S_2]+ (1/2)
\Pr[S_1=S_2]$ is monotonically increasing in $t$. Consider the
random variable $Z=S_1-S_2$. At each phase it increases by $+1$ with
probability $\mu_1(1-\mu_2)$, decreases by $-1$ with probability
$(1-\mu_1)\mu_2$ and otherwise does not change.

Consider the values of $Z$ up to phase $t$. We really care only
about the probability that is shifted from positive to negative and
vice versa.

First, consider the probability that $Z=0$. We can partition it to
$S_1=S_2=r$ events, and let $p(r,r)$ be the probability of this
event. For each such event, we have $p(r,r)\mu_1$ moved to $Z=+1$
and $p(r,r)\mu_2$ moved to $Z=-1$. Since $\mu_1>\mu_2$ we have that
$p(r,r)\mu_1\geq p(r,r)\mu_2$ (note that $p(r,r)$ might be zero, so
we do not have a strict inequality).

Second, consider the probability that $Z=+1$ or $Z=-1$. We can
partition it to $S_1=r+1;S_2=r$ and $S_1=r;S_2=r+1$ events, and let
$p(r+1,r)$ and $p(r,r+1)$ be the probabilities of those events.
It is not hard to see that $p(r+1,r)\mu_2=p(r,r+1)\mu_1$.
This implies that the probability mass moved from $Z=+1$ to $Z=0$ is
identical to that moved from $Z=-1$ to $Z=0$.

We have showed that $\Pr[S_1>S_2]+ (1/2) \Pr[S_1=S_2]$ and therefore
the expected valued of the exploit action is non-decreasing. Since
we have that the size of the phases are increasing, the $\BIR$ is
strictly increasing between phases and identical within each phase.

We now analyze the $\BIR$ regret. Note that agent $n$ is in phase
$O(n^{2/3})$ and the length of his phase is $O(n^{1/3})$. The $\BIR$
has two parts. The first is due to the exploration, which is at most
$O(n^{-1/3})$. The second is due to the probability that we exploit
the wrong action. This happens with probability $\Pr[S_1<S_2]+ (1/2)
\Pr[S_1=S_2]$ which we can bound using a Chernoff bound by
$e^{-O(\Delta^2n^{2/3})}$, since we explored each action
$O(n^{2/3})$ times.
\end{proof}

\begin{remark}
Actually we have a tradeoff depending on the parameter $m_t$ between
the regret due to exploration and exploitation. (Note that the
monotonicity is always guarantee assuming $m_t$ is monotone.) If we
can set that $m_t = 2^t$ then at time $n$ we have $2/ n$ probability
of an exploit action. For the explore action we are in phase $\log
n$ so the probability of a sub-optimal explore action is
$n^{-O(\Delta^{-2})}$. This should give us
$\BIR(n)=O(n^{-O(\Delta^{-2})})$.
\end{remark}

\subsection{``Smart" MAB algorithms that combine exploration and exploitation}
\label{sec:MAB-smart}

MAB algorithm \SuccesiveEliminationReset works as follows. It keeps
a set of surviving actions $A_s\subseteq A$, where initially
$A_s=A$. The agents are partition into phases, where each phase is a
random permutation of the non-eliminated actions.
Let $\hat{\mu}_{i,t}$ be the average of the rewards of action $i$ up
to phase $t$ and $\hat{\mu}^*=\max_i \hat{\mu}_{i,t}$. We eliminate
action $i$ at the end of phase $t$, i.e., delete it from $A_s$, if
$\hat{\mu}_t^*-\hat{\mu}_{i,t} > \log(T/\delta)/\sqrt{t}$.
In \SuccesiveEliminationReset we simply reset the algorithm with
$A=A_s-A_{e,t}$, where $A_{e,t}$ is the set of eliminated actions
after phase $t$. Namely, we restart $\hat{\mu}_{i,t}$ and ignore the
old rewards before the elimination.

\begin{lemma}
The algorithm \SuccesiveEliminationReset, has, with probability
$1-\delta$, \BIR$(n)=O(\log(T/\delta)/\sqrt{n/K})$.
\end{lemma}

\begin{proof}
Let the best action be $a^*=\arg\max_a \mu_a$. With probability
$1-\delta$ at any time $n$ we have that for any action $i\in A_s$
that $|\hat{\mu}_i -\mu_i|\leq \log(T/\delta)/\sqrt{n/K}$, and
$a^*\in A_s$. This implies that any action $a$ such
$\mu_{a^*}-\mu_{a}> 3\log(T/\delta)/\sqrt{n/K}$ is eliminated.
Therefore, any action in $A_s$ has \BIR$(n)$ of at most
$6\log(T/\delta)/\sqrt{n/K}$.
\end{proof}

\begin{lemma}
Assume that if $\mu_i\geq \mu_j$ then the rewards $r_i$
stochastically dominates the rewards $r_j$.
Then, \SuccesiveEliminationReset is monotone
\end{lemma}

\begin{proof}
Consider the first time $T$ an action is eliminated, and let
$T=\tau$ be a realized value of $T$. Then, clearly for $n<\tau$ we
have $\rew(n)=\rew(1)$ .

Consider two actions $a_1,a_2\in A$, such that $\mu_{a_1} \geq
\mu_{a_2}$. At time $T=\tau$, the probability that  $a_1$ is
eliminated is smaller than the probability that $a_2$ is eliminated.
This follows since $\hat{\mu}_{a_1}$ stochastically dominates
$\hat{\mu}_{a_2}$, which implies that for any threshold $\theta$ we
have $\Pr[\hat{\mu}_{a_1}\geq\theta]\geq
\Pr[\hat{\mu}_{a_2}\geq\theta]$.

After the elimination we consider the expected reward of the
eliminated action $\sum_{i\in A} \mu_i q_i$, where $q_i$ is the
probability that action $i$ was eliminated in time $T=\tau$. We have
that $q_i \leq q_{i+1}$, from the probabilities of elimination.

The sum $\sum_{i\in A} \mu_i q_i$ with $q_i \leq q_{i+1}$ and
$\sum_i q_i=1$ is maximized by setting $q_i=1/|A|$. (We can see that
if there are $q_i\neq 1/|A|$, then there are two $q_{i}< q_{i+1}$,
and one can see that setting both to $(q_{i}+ q_{i+1})/2$ increases
the value.) Therefore we have that the $\rew(\tau)\geq\rew(\tau-1)$.

Now we can continue by induction. For the induction, we can show the
property for {\em any} remaining set of at most $k-1$ actions. The
main issue is that \SuccesiveEliminationReset restarts from scratch,
so we can use induction.
\end{proof}

\section{Non-degeneracy via a random perturbation}
\label{app:perturb}
We show that Assumption~\eqref{eq:assn-distinct} holds almost surely under a small random perturbation of the prior. We focus on problem instances with 0-1 rewards, and assume that the prior $\priorMu$ is independent across arms and has a finite support.%
\footnote{The assumption of 0-1 rewards is for clarity. Our results hold under a more general assumption that for each arm $a$, rewards can only take finitely many values, and each of these values is possible (with positive probability) for every possible value of the mean reward $\mu_a$.}
 Consider the probability vector in the prior for arm $a$:
\[ \vec{p}_a = \left(\; \Pr[\mu_a=\nu]:\; \nu\in \support(\mu_a)\; \right).\]
We apply a small random perturbation independently to each such vector:
\begin{align}\label{eq:app:perturb:noise}
\vec{p}_a \leftarrow \vec{p}_a + \vec{q}_a,
    \quad \text{where}\quad \vec{q}_a\sim  \mN_a.
\end{align}
Here $\mN_a$ is the noise distribution for arm $a$: a distribution over real-valued, zero-sum vectors of dimension $d_a = |\support(\mu_a)|$. We need the noise distribution to satisfy the following property:
\begin{align}\label{eq:app:perturb:noise-prop}
\forall x\in [-1,1]^{d_a}\setminus \{0\}\qquad
\Pr_{q\sim \mN_a}\left[ x\cdot(\vec{p}_a+ q) \neq 0 \right] =1.
\end{align}

\begin{theorem}\label{thm:perturb}
Consider an instance of MAB with 0-1 rewards. Assume that the prior $\priorMu$ is independent across arms, and each mean reward $\mu_a$ has a finite support that does not include $0$ or $1$. Assume that noise distributions $\mN_a$ satisfy property \eqref{eq:app:perturb:noise-prop}. If random perturbation~\eqref{eq:app:perturb:noise} is applied independently to each arm $a$, then \refeq{eq:assn-distinct} holds almost surely for each history $h$.
\end{theorem}

\begin{remark}
As a generic example of a noise distribution which satisfies
  Property \eqref{eq:app:perturb:noise-prop}, consider the uniform
  distribution $\mN$ over the bounded convex set
  \[ Q = \left\{q \in \R^{d_a} \mid q \cdot \vec{1} = 0 \mbox{ and }
      \|q\|_2 \leq \eps\right\}, \] where $\vec{1}$ denotes the all-1
  vector. If $x = a \vec{1}$ for some non-zero value of $a$, then
  \eqref{eq:app:perturb:noise-prop} holds because
  $$x \cdot (p+q) = x \cdot p = a\neq 0.$$ Otherwise, denote
  $p=\vec{p}_a$ and observe that $x\cdot({p}+ {q}) = 0$ only if
  $x \cdot q = c \triangleq x \cdot (-p)$. Since $x\neq \vec{1}$, the
  intersection $Q\cap\{ x\cdot q = c \}$ either is empty or has
  measure 0 in $Q$, which implies
  $\Pr_{{q}}\left[ x\cdot({p}+ {q}) \neq 0 \right] =1$.
\end{remark}

To prove Theorem~\ref{thm:perturb}, it suffices to focus on two arms, and perturb one of them. Since realized rewards have finite support, there are only finitely many possible histories. Therefore, it suffices to focus on a fixed history $h$.

\begin{lemma}\label{lm:perturb}
Consider an instance of MAB with 0-1 rewards. Assume that the prior $\priorMu$ is independent across arms, and that $\support(\mu_1)$ is finite and does not include $0$ or $1$. Fix history $h$. Suppose random perturbation~\eqref{eq:app:perturb:noise} is applied to arm $1$, with noise distribution $\mN_1$ that satisfies \eqref{eq:app:perturb:noise-prop}. Then
    $\E[\mu_1\mid h] \neq \E[\mu_2\mid h] $
almost surely.
\end{lemma}

\begin{proof}
Note that $\E[\mu_a\mid h]$ does not depend on the algorithm which produced this history. Therefore, for the sake of the analysis, we can assume w.l.o.g. that this history has been generated by a particular algorithm, as long as this algorithm can can produce this history with non-zero probability. Let us consider  the algorithm that deterministically chooses same actions as $h$.

Let $S = \support(\mu_1)$. Then:
\begin{align*}
\E[\mu_1\mid h]
    &= \sum_{\nu\in S} \nu \cdot \Pr[\mu_1 =\nu \mid h]
    = \sum_{\nu\in S} \nu \cdot \Pr[h \mid \mu_1 =\nu] \cdot \Pr[\mu_1=\nu]\;/\;\Pr[h], \\
\Pr[h]
    &= \sum_{\nu\in S} \Pr[h \mid \mu_1 =\nu] \cdot \Pr[\mu_1=\nu].
\end{align*}
Therefore,
    $\E[\mu_1\mid h] = \E[\mu_2\mid h] $
if and only if
\begin{align*}
\sum_{\nu\in S} (\nu-C) \cdot \Pr[h \mid \mu_1 =\nu] \cdot \Pr[\mu_1=\nu] = 0,
\quad\text{ where }\quad C=\E[\mu_2\mid h].
\end{align*}
Since $\E[\mu_2\mid h]$ and $\Pr[h \mid \mu_1 =\nu]$ do not depend on the probability vector $\vec{p}_1$, we conclude that
\begin{align*}
  \E[\mu_1\mid h] = \E[\mu_2\mid h]
\quad\Leftrightarrow\quad x\cdot \vec{p}_1 =0,
\end{align*}
where vector
    \[ x := \left(\; (\nu-C) \cdot \Pr[h \mid \mu_1 =\nu]:\; \nu\in S\; \right) \in [-1,1]^{d_1}\]
does not depend on $\vec{p}_1$.

Thus, it suffices to prove that $x\cdot \vec{p}_1 \neq 0$ almost surely under the perturbation. In a formula:
\begin{align}\label{eq:lm:perturb-pf}
    \Pr_{q\sim \mN_1}\left[ x\cdot( \vec{p}_1+q) \neq 0 \right] =1
\end{align}

Note that $\Pr[h \mid \mu_1 =\nu]>0$ for all $\nu\in S$, because $0,1\not\in S$. It follows that at most one coordinate of $x$ can be zero. So  \eqref{eq:lm:perturb-pf} follows from property \eqref{eq:app:perturb:noise-prop}.
\end{proof}

\end{document}